\documentclass[journal]{IEEEtran}
	
\usepackage{caption}
\usepackage{subcaption}
\usepackage{mathrsfs}
\usepackage[cmex10]{amsmath}
\usepackage{cite,graphicx,epsfig,wrapfig,amsfonts,amssymb,algorithmic,threeparttable,color,url}
\usepackage{mdwmath,multirow}
\usepackage{mdwtab}
\usepackage[ruled,vlined]{algorithm2e}
\usepackage{amsthm}

\newtheorem{theorem}{Theorem}
\newtheorem{lemma}{Lemma}

\newtheorem{corollary}{Corollary}

\theoremstyle{definition}

\graphicspath{{./}{figures/}}

\ifCLASSINFOpdf

\else

\fi

\begin{document}
\title{A Traffic Load Balancing Framework for Software-defined Radio Access Networks Powered by Hybrid Energy Sources}
\author{\IEEEauthorblockN{Tao Han, \emph{{Student Member, IEEE}}, and
Nirwan Ansari, \emph{{Fellow, IEEE}}}\\
\IEEEauthorblockA{Advanced Networking Laboratory \\
Department of Electrical and Computer Engineering \\
New Jersey Institute of Technology, Newark, NJ, 07102, USA\\
Email: \{th36, nirwan.ansari\}@njit.edu}
\thanks{This work was supported in part by NSF under grant no. CNS-1218181 and no. CNS-1320468.}}
\maketitle
%---- page number related --- %
%\pagestyle{empty}
%\thispagestyle{empty}
%----- end of page number related ----%
\begin{abstract}
Dramatic mobile data traffic growth has spurred a dense deployment of small cell base stations (SCBSs). Small cells enhance the spectrum efficiency and thus enlarge the capacity of mobile networks. Although SCBSs consume much less power than macro BSs (MBSs) do, the overall power consumption of a large number of SCBSs is phenomenal. As the energy harvesting technology advances, base stations (BSs) can be powered by green energy to alleviate the on-grid power consumption. For mobile networks with high BS density, traffic load balancing is critical in order to exploit the capacity of SCBSs. To fully utilize harvested energy, it is desirable to incorporate the green energy utilization as a performance metric in traffic load balancing strategies. In this paper, we have proposed a traffic load balancing framework that strives a balance between network utilities, e.g., the average traffic delivery latency, and the green energy utilization. Various properties of the proposed framework have been derived. Leveraging the software-defined radio access network architecture, the proposed scheme is implemented as a virtually distributed algorithm, which significantly reduces the communication overheads between users and BSs. The simulation results show that the proposed traffic load balancing framework enables an adjustable trade-off between the on-grid power consumption and the average traffic delivery latency, and saves a considerable amount of on-grid power, e.g., 30\%, at a cost of only a small increase, e.g., 8\%, of the average traffic delivery latency.

\end{abstract}
\IEEEpeerreviewmaketitle

\section{Introduction}
\label{sec:introduction}
Proliferation of wireless devices and bandwidth greedy applications drive the exponential growth of mobile data traffic that leads to a continuous surge in capacity demands across mobile networks. Heterogeneous network (HetNet) is one of the key technologies for enhancing mobile network capacity to satisfy the capacity demands \cite{Andrews:2014:AOLB}. In HetNet, low-power base stations referred to as small cell base stations (SCBSs) are densely deployed to enhance the spectrum efficiency of the network and thus increase the network capacity. Owing to the disparate transmit powers and base station (BS) capabilities, traditional user association metrics such as the signal-to-interference-plus-noise ratio (SINR) and the received-signal-strength-indication (RSSI) may lead to a severe traffic load imbalance \cite{Andrews:2014:AOLB}. Hence, user association algorithms should be well designed to balance traffic loads and thus to fully exploit the capacity potential of HetNet.

In order to maximize network utilities, balancing traffic loads requires coordination among BSs. The dense deployment of BSs in HetNet increases the difficulty on coordinating BSs. To address this issue, software-define radio access network (SoftRAN) architecture \cite{Gudipati:2013:SSD} has been proposed. SoftRAN enables coordinated radio resource management in the centralized control plane with a global view of network resources and traffic loads. The user association algorithm leveraging the SoftRAN architecture is desired for future mobile networks with an extremely dense BS deployment.

Owing to the direct impact of greenhouse gases on the earth environment and the climate change, the energy consumption of Information and Communications Technology (ICT) is becoming an environmental and thus social and economic issue. Mobile networks are among the major energy hoggers of communication networks, and their contributions to the global energy consumption increase rapidly. Therefore, greening mobile networks is crucial to reducing the carbon footprints of ICT. Although SCBSs consume less power than macro BSs (MBSs), the number of SCBSs will be orders of magnitude larger than that of MBSs for a wide scale network deployment. Hence, the overall power consumption of such a large number of SCBSs will be phenomenal. Greening HetNets have thus attracted tremendous research efforts \cite{Han:2012:OGC,Hasan:2011:GCN}.

As energy harvesting technologies advance, green energy such as sustainable biofuels, solar and wind energy can be utilized to power BSs  \cite{Han:2014:PMN}. Telecommunication companies such as Ericsson and Nokia Siemens have designed green energy powered BSs for mobile networks \cite{Ericson:2007:SEU}. By adopting green energy powered BSs, mobile network operators (MNOs) may further save on-grid power consumption and thus reduce their $CO_{2}$ emissions. However, since the green energy generation is not stable, green energy may not be a reliable energy source for mobile networks. Therefore, future mobile networks are likely to adopt hybrid energy supplies: on-grid power and green energy. Green energy is utilized to reduce the on-grid power consumption and thus reduce the $CO_{2}$ emissions while on-grid power is utilized as a backup power source.

In HetNets with hybrid energy supplies, the utilization of green energy should be integrated into user association metrics to optimize the green energy usage. For instance, while balancing traffic loads, MNOs may enable BSs with sufficient green energy to serve more traffic loads while reducing the traffic loads of BSs consuming on-grid power \cite{Han:2013:OOG}. The traffic load balancing with the consideration of green energy may not maximize network utilities such as the network capacity and the traffic delivery latency. Therefore, a trade-off between the green energy utilization and network utilities should be carefully evaluated in balancing traffic loads among BSs. In addition, as a result of the trade-off, users' utilities such as data rates and the service latency may be decreased because of the consideration of green energy in the traffic load balancing. Thus, users may not cooperate in the traffic load balancing. For example, a distributed user association algorithm may involve multiple interactions between users and BSs and require users to report their measurements to BSs \cite{Ye:2013:UAL,Han:2013:GALA}. Seeking to improve their own utilities, they may not report the correct information to BSs. Therefore, it is desirable to hide BSs' energy information from users to avoid counterfeit reports.

In this paper, we propose a virtually distributed user association scheme that leverages the SoftRAN concept. We generate virtual users and virtual BSs (vBSs) in the radio access networks controller (RANC) to emulate a distributed user association solution that requires iterative user association adjustments between users and BSs. This scheme runs the user association optimization in the RANC, and thus significantly reduces the communication overhead over the air interface. In this scheme, users report their downlink data rates calculated based on perceived SINRs via an associating BS to the RANC where traffic loads from individual users and BSs are measured. The RANC optimizes the BS operation status that reflects the price for a user to access a BS. The user association is determined by the BS operation status and the users' downlink data rates. The proposed scheme, in determining user association, allows an adaptable trade-off between network utilities, e.g., the average traffic delivery latency and the green energy utilization. Meanwhile, running the user association within the RANC avoids leaking energy information to users. As a result, users have no obvious incentives to counterfeit reports. Based on the above features, we name the proposed user association scheme as vGALA: \textbf{v}irtualized \textbf{G}reen energy \textbf{A}ware and \textbf{L}atency \textbf{A}ware user association \footnote{The initial idea about green energy aware and latency aware user association was presented at GLOBECOM 2013 \cite{Han:2013:GALA}.}.

The rest of the paper is organized as follows. In Section \ref{sec:related_work}, we briefly review related works. In Section \ref{sec:sys_model}, we define the system model and formulate the user association problem. Section \ref{sec:distributed_scheme} presents the vGALA scheme. Section \ref{sec:tradeoff_admission} discusses the practicality of the vGALA scheme. Section \ref{sec:simulation} shows the simulation results, and concluding remarks are presented in Section \ref{sec:conclusion}.

\section{Related Works}
\label{sec:related_work}
Balancing traffic loads in HetNet has been extensively studied in recent years \cite{Wang:2013:MMN}. In mobile networks, traffic loads among BSs is balanced by executing handover procedures. In the LTE system, there are three types of handover procedures: Intra-LTE handover, Inter-LTE handover, and Inter-RAT (radio access technology) handover \cite{LTE:2013:handover}. There are two ways to trigger handover procedures. The first one is ``Network Evaluated" in which the network triggers handover procedures and makes handover decisions. The other one is ``Mobile Evaluated" in which a user triggers the handover procedure and informs the network about the handover decision. Based on the radio resource status, the network decides whether to approve the user's handover request. In 4G and LTE networks, a hybrid approach is usually implemented where a user measures parameters of the neighboring cells and reports the results to the network. The network makes the handover decision based on the measurements. Here, the network can decide which parameters should be measured by users.

Aligning with the above procedures, various traffic load balancing algorithms have been proposed to optimize the network utilities. The most practical traffic load balancing approach is the cell range expansion (CRE) technique that biases users' receiving SINRs or data rates from some BSs to prioritize these BSs in associating with users \cite{Damn:2011:S3GPP}. Owing to the transmit power difference between MBSs and SCBSs, a large bias is usually given to SCBSs to offload users to small cells \cite{Andrews:2014:AOLB}. By applying CRE, a user associates with the BS from which the user receives the maximum biased SINR or data rate. Although CRE is simple, it is challenging to derive the optimal bias for BSs. Singh \emph{et al.} \cite{Singh:2013:OHN} provided a comprehensive analysis on traffic load balancing using CRE in HetNet. The authors investigated the selection of the bias value and its impact on the SINR coverage and the downlink rate distribution in HetNet.

The traffic load balancing problem can also be modeled as an optimization problem and solved by convex optimization approaches. Ye \emph{et al.} \cite{Ye:2013:UAL} modeled the traffic load balancing problem as a utility maximization problem and developed distributed user association algorithms based on the primal-dual decomposition. Kim \emph{et al.} \cite{Kim:2012:DOU} proposed an $\alpha$-optimal user association algorithm to achieve flow level load balancing under spatially heterogeneous traffic distribution. The proposed algorithm may maximize different network utilities, e.g., the traffic latency and the network throughput, by properly setting the value of $\alpha$. In addition, game theory has been exploited to model and solve the traffic load balancing problems. Aryafar \emph{et al.} \cite{Aryafar:2013:RSG} modeled the traffic load balancing problem as a congestion game in which users are the players and user association decisions are the actions.

The above solutions, though effectively balance the traffic loads to maximize the network utilities, do not consider the green energy utilization as a performance metric in balancing traffic loads. As green energy technologies advance, powering BSs with green energy is a promising solution to save on-grid power and reduce the carbon footprints \cite{Han:2014:PMN}. It is desirable to recognize green energy as one of the performance metrics when balancing the traffic loads. Zhou \emph{et al.} \cite{Zhou:2010:ESA:} proposed a handover parameter tuning algorithm for target cell selection, and a power control algorithm for coverage optimization to guide mobile users to access the BSs with renewable energy supply. Considering a mobile network powered by multiple energy sources, Han and Ansari \cite{Han:2013:OOG} proposed to optimize the utilization of green energy for cellular networks by optimizing BSs' transmit powers. The proposed algorithm achieves significant on-grid power savings by scheduling the green energy consumption along the time domain for individual BSs, and balancing the green energy consumption among BSs. The authors have also proposed a user association algorithm that jointly optimizes the average traffic delivery latency and the green energy utilization \cite{Han:2013:GALA}.
%Jo \emph{et al.} \cite{Jo:2012:HCN} proposed data rate biasing algorithms to balance traffic loads among MBSs and PBSs. The data rate biasing algorithms perform user-BS association according to the biased measured pilot signal strength, and enables the traffic to be offloaded from MBSs to PBSs.
%Corroy \emph{et al.} \cite{Corroy:2012:DCA} proposed a dynamic user-BS association algorithm to maximize the sum rate of the network and adopted data rate biasing to balance the traffic load among BSs.

\section{System Model and Problem Formulation}
\label{sec:sys_model}
In this paper, we consider a HetNet with multiple MBSs and SCBSs as shown in Fig. \ref{fig:heter_net_arch}. Both the MBSs and SCBSs are powered by on-grid power and green energy. We consider solar power as the green energy source. We focus on balancing the downlink traffic loads among BSs by designing the green energy and latency aware user association scheme. We adopt a software-defined radio access network (SoftRAN) architecture in which all BSs are controlled by the RAN controller (RANC). The RANC has a global view of BSs' traffic loads and green energy. The user association is optimized by the RANC. The specific design of the RANC is beyond the scope of this paper.
\begin{figure}[htb]
\centering
\includegraphics[scale=0.25]{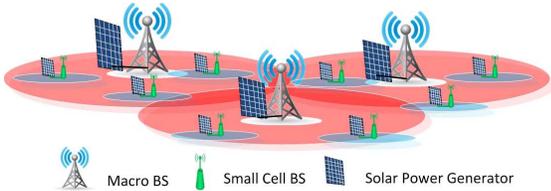}
\caption{A HetNet powered by hybrid energy sources: on-grid power and green energy.}
\label{fig:heter_net_arch}
\vspace{-20pt}
\end{figure}

\subsection{Traffic model}
%We consider MBSs and SCBSs are deployed to provide data communications to an area.
Denote $\mathcal{B}$ as a set of BSs including both the MBS and SCBSs. We assume that the traffic arrives according to a Poisson point process with the average arrival rate per unit area at location $x$ equaling to $\lambda(x)$, and the traffic size (packet size) per arrival has a general distribution with the average traffic size of $\nu(x)$. Assuming a mobile user at location $x$ is associated with the $j$th BS, then the user's downlink data rate $r_{j}(x)$ that will end up becoming available to the user can be generally expressed as a logarithmic function of the perceived SINR, $SINR_{j}(x)$, according to the Shannon-Hartley theorem~\cite{Kim:2012:DOU},
\begin{equation}
\label{eq:user_rate}
r_{j}(x)=W_{j}log_{2}(1+SINR_{j}(x)),
\end{equation}
where $W_{j}$ is the total bandwidth in the $j$th BS.
\begin{equation}
\label{eq:user_SINR}
SINR_{j}(x)=\frac{P_{j}g_{j}(x)}{\sigma^{2}+\sum_{k \in \mathcal{I}_{j}}I_{k}(x)}.
\end{equation}
Here, $P_{j}$ is the transmission power of the $j$th BS, $\mathcal{I}_{j}$ represents the set of interfering BSs which is defined as the set of BSs whose transmission interferes the $j$th BS's transmission toward a user at location $x$, $I_{k}(x)$ is the average interference power seen by a user at location $x$ from the $k$th BS, $\sigma^{2}$ denotes the noise power level and $g_{j}(x)$ is the channel gain between the $j$th BS and the user at location $x$. Here, the channel gain reflects only the slow fading including the path loss and the shadowing. We assume the channel gain is measured at a large time scale, and thus fast fading is not considered.

In HetNet, the total bandwidth in a BS is determined by the network's frequency planning. Different frequency reuse strategies result in different inter-BS interference. In this paper, we assume the network's frequency reuse strategy is given and static. Thus, $\mathcal{I}_{j}$ contains the set of BSs who share the same spectrum with the $j$th BS. We assume users experience a roughly static interference from the interfering BSs. Although the inter-BS interference in HetNet varies depending on the activities in the interfering BSs, the interference can be well coordinated via time domain techniques, frequency domain techniques and power control techniques \cite{Lopez:2011:EICIC}. Therefore, the inter-BS interference can be reasonably modeled as a static value for analytical simplicity. The static inter-BS interference model has also been adopted in previous works for modeling the user association problem \cite{Son:2009:DAL, Kim:2012:DOU}.
%\end{remark}

The average traffic load density at location $x$ in the $j$th BS is
\begin{equation}
\label{eq:bs_point_load}
\varrho_{j}(x)=\frac{\lambda(x)\nu(x)\eta_{j}(x)}{r_{j}(x)}
\end{equation}
Here, $\eta_{j}(x)$ is an indicator function. If $\eta_{j}(x)=1$, the user at location $x$ is associated with the $j$th BS; otherwise, the user is not associated with the $j$th BS.
Assuming mobile users are uniformly distributed in the area and denoting $\mathcal{A}$ as the coverage area of all the BSs, based on Eq.~(\ref{eq:bs_point_load}), we derive the average traffic loads in the $j$th BS expressed as
\begin{equation}
\label{eq:bs_load}
\rho_{j}=\int_{x \in \mathcal{A}}\varrho_{j}(x)dx.
\end{equation}
%\rho_{j}=\frac{1}{\mathcal{A}_{j}(k)}\int_{x \in \mathcal{A}}\varrho_{j}(x)dx.
The value of $\rho_{j}$ indicates the fraction of time during which the $j$th BS is busy.

We assume that traffic arrival processes at individual locations are independent. Since the traffic arrival per unit area is a Poisson point process, the traffic arrival in the $j$th BS, which is the sum of the traffic arrivals in its coverage area, is a Poisson process. The required service time per traffic arrival for a user at location $x$ in the $j$th BS is $\gamma_{j}= \frac{\nu(x)}{r_{j}(x)}$.
%\begin{equation}
%\label{eq:required_time}
%\gamma_{j}= \frac{\nu(x)}{r_{j}(x)}.
%\end{equation}
Since $\nu(x)$ is the average traffic size per arrival which follows a general distribution, the user's required service time is also a general distribution. Hence, a BS's service rate follows a general distribution. Therefore, a BS's downlink transmission process realizes a M/G/1 processor sharing queue, in which multiple users share the BS's downlink radio resource \cite{Kleinrock:1976:QS}.

In mobile networks, various downlink scheduling algorithms have been proposed to enable proper sharing of the limited radio resource in a BS \cite{Capozzi:2013:DPS}. These algorithms are designed to maximize the network capacity, enhance the fairness among users, or provision QoS services. According to the scheduling algorithm, users are assigned different priorities on sharing the downlink radio resource. As a result, users in different priority groups perceive different average waiting time. Since traffic arrives at a BS according to Possion arrival statistics, the allowed variation in the average waiting times among different priority groups is constrained by the Conservation Law \cite{Kleinrock:1976:QS}. The integral constraint on the average waiting time in the $j$th BS can be expressed as
\begin{equation}
\label{eq:conservation_law}
\bar{L}_{j}=\frac{\rho_{j}E(\gamma_{j}^{2})}{2(1-\rho_{j})}.
\end{equation}
This indicates that given the users' required service time in the $j$th BS, if the scheduling algorithm gives some users higher priority and reduces their average waiting time, it will increase the average waiting time of the other users. Therefore, $\bar{L}_{j}$ generally reflects the $j$th BS's performance in terms of users' average waiting time. 
Since $E(\gamma_{j}^{2})$ mainly reflects the traffic characteristics, we assume that $E(\gamma_{j}^{2})$ is roughly constant during a user association process and thus $\vartheta_{j}=\frac{E(\gamma_{j}^{2})}{2}$ can be considered as a constant. Thus, we adopt
\begin{equation}
\label{eq:latency_indicator}
L(\rho_{j})=\frac{\vartheta_{j}\rho_{j}}{1-\rho_{j}}
\end{equation}
as a general latency indicator for the $j$th BS. A smaller $L(\rho_{j})$ indicates that the $j$th BS introduces less latency to its associated users. Therefore, we use $L(\rho_{j})$ to reflect the $j$th BS's average traffic delivery latency.
%the BS realizes an M/G/1 queuing system. Assuming mobile users are served based on the round robin fashion, the traffic delivery in the BS can be modeled as an $M/G/1-PS$ (processor sharing) queue \cite{Kleinrock:1976:QS}.
%\begin{remark}
%any convex function of $\rho_{j}$ fits our algorithm
%\end{remark}
\subsection{Energy model}
%\begin{figure}
%\centering
%\includegraphics[scale=0.3]{/pic/solar_data_fig.eps}
%\caption{The solar power data \cite{UCSD:solar_data}.}
%\label{fig:solar_data}
%\end{figure}
\begin{figure}
\centering
\includegraphics[scale=0.3]{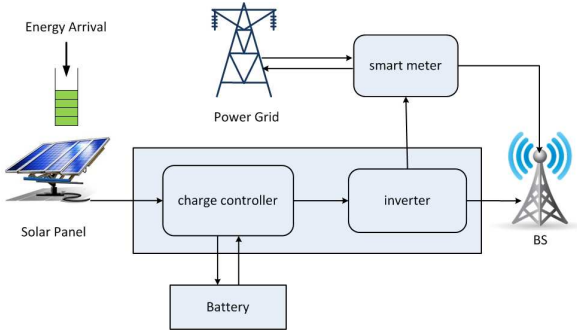}
\caption{A hybrid energy powered BS.}
\label{fig:green_bs}
\vspace{-16pt}
\end{figure}
In the network, both MBSs and SCBSs have their own solar panels for generating green energy. Therefore, BSs are powered by hybrid energy sources: on-grid power and green energy. If green energy generated by its solar panel is not sufficient, the BS consumes on-grid power. Since MBSs usually consume more energy than SCBSs, we assume that MBSs are equipped with larger solar panels that have a higher energy generation capacity than that of a SCBS. A reference design of a hybrid energy powered BS \cite{Han:2014:PMN} is shown in Fig \ref{fig:green_bs}. The charge controller optimizes the green energy utilization based on the solar power intensity, the power consumption of BSs, and energy prices on power grid. Here, the green energy utilization is optimized over time horizon. For example, the charge controller may predict the solar power intensity and mobile traffic loads in a BS over a certain period of time, e.g., 24 hours. The prediction can be based on statistical data and real time weather forecasts. The charge controller according to the prediction determines how much green energy should be utilized to power a BS during a specific time period, e.g., the time duration between two consecutive traffic load balancing procedures.

In this paper, instead of investigating how to optimize the green energy utilization over the time horizon, we aim to study how to balance traffic loads among BSs to save on-grid energy within the duration of a traffic balancing procedure. Therefore, we assume that the amount of available green energy for powering a BS is a constant within this duration as determined by the charge controller. It is reasonable to assume that the available green energy is constant because the traffic load balancing process is at a time scale of several minutes \cite{Kim:2012:DOU} while solar power generation is usually modeled at a time scale of an hour \cite{Farbod:2007:RAO}. Denote $e_{j}$ as the amount of green energy for powering the $j$th BS in a traffic load balancing procedure. If the power consumption of the $j$th BS is larger than $e_{j}$, the BS consumes on-grid power. Otherwise, the residual green energy will be either stored in battery for future usage or uploaded to power grid via the smart meter. Since we are not focusing on optimizing the green energy utilization over the time horizon, we simply model the BS's on-grid energy consumption is zero when the BS's power consumption is less than $e_{j}$. In other words, we do not consider the redistribution of the residual green energy in our model.

The BS's power consumption consists of two parts: the static power consumption and the dynamic power consumption~\cite{Auer:2011:HME}. The static power consumption is the power consumption of a BS without carrying any traffic load. The dynamic power consumption refers to the additional power consumption caused by traffic loads in the BS, which can be well approximated by a linear function of the traffic loads~\cite{Auer:2011:HME}. Denote $p^{s}_{j}$ as the static power consumption of the $j$th BS. Then, the $j$th BS's power consumption can be expressed as
\begin{equation}
\label{eq:bs_power_consumption}
p_{j}=\beta_{j}\rho_{j}+p^{s}_{j}.
\end{equation}
Here, $\beta_{j}$ is the load-power coefficient that reflects the relationship between the traffic loads and the dynamic power consumption in the $j$th BS. The BS power consumption model can be adjusted to model the power consumption of either MBSs or SCBSs by incorporating and tweaking the static power consumption and the load-power
coefficient. The on-grid power consumption in the $j$th BS is
\begin{equation}
\label{eq:bs_grid_energy}
p^{o}_{j}=\max{(p_{j}-e_{j},0)}.
\end{equation}

\subsection{Problem formulation}
In determining the user association, the network aims to strive for a trade-off between network utilities, e.g., the average traffic delivery latency and the on-grid power consumption. In this paper, we focus on designing a user association algorithm to enhance the network performance by reducing the average traffic delivery latency in BSs as well as to reduce the on-grid power consumption by optimizing the green energy usage.

On the one hand, to reduce the average traffic delivery latency, the network desires to minimize the summation of the latency indicators of BSs. On the other hand, since BSs are powered by both green energy and on-grid power, the network seeks to minimize the usage of on-grid power by optimizing the utilization of green energy.
According to Eq.~(\ref{eq:bs_grid_energy}), on-grid power is only consumed when green energy is not sufficient in the BS. When $p_{j}>e_{j}$, to alleviate on-grid power consumption, the $j$th BS has to reduce its traffic loads. We define the green traffic capacity as the maximum traffic loads that can be supported by green energy. Denote $\hat{\rho}_j$ as the green traffic capacity of the $j$th BS. Then,
\begin{equation}
\label{eq:green_load}
\hat{\rho}_{j}=\max{(\epsilon,\min(\frac{e_{j}-p^{s}_{j}}{\beta_{j}},1-\epsilon))}.
\end{equation}
Here, $\epsilon$ is an arbitrary small positive constant to guarantee $0<\hat{\rho}_{j}<1$.
To reduce traffic loads from $\rho_{j}$ to $\hat{\rho}_{j}$, the $j$th BS has to shrink its coverage area. As a result, its traffic loads are offloaded to its neighboring BSs and may lead to traffic congestion in the neighboring BSs. The traffic congestion increases the average traffic delivery latency of the network. To achieve a trade-off between the average traffic delivery latency and the on-grid power consumption, we define the energy-latency coefficient in the $j$th BS as $\theta_{j}$. We further define the desired traffic loads in the $j$th BS after the energy-latency trade-off as
\begin{equation}
\label{eq:tradeoff_load}
\tau_{j}=(1-\theta_{j})\rho_{j}+\theta_{j}\hat{\rho}_{j}.
\end{equation}
Here, $0\leq \theta_{j} \leq 1$.
If $\theta_{j}$ is set to zero, the $j$th BS's desired traffic loads are its actual traffic loads without considering green energy. In this case, we consider the $j$th BS being latency-sensitive; otherwise, if $\theta_{j}$ equal to one, the $j$th BS's desired traffic loads are dominated by its green traffic capacity and thus the BS is energy-sensitive. The selection of $\theta_{j}$ reflects the $j$th BS's energy-latency trade off that will be discussed in Section \ref{subsec:energy_latency_trade}. We assume $\theta_{j}$ remains constant within the duration of a user association process.
%Since $\theta_{j}$ is determined by the network operator and may not changes frequently, we assume $\theta_{j}$ keeps constant within a duration of a user association.

Since mobile devices are powered by battery, it is desirable to guarantee the energy efficiency of mobile devices while performing the traffic load balancing~\cite{Raj:2013:EAM}. To ensure the energy efficiency of mobile devices, we restrict a user to only associate with the BSs to which the user's uplink pathloss is smaller than a predefined threshold.
Considering all the above factors, the user association (UA) problem is formulated as
\begin{eqnarray}
\label{eq:object_network_goal}
\min_{\boldsymbol{\rho}} && \sum_{j \in \mathcal{B}}w_{j}(\rho_{j})L(\rho_{j})\\
\label{eq:constraint_omge}
subject\; to: && 0\leq\rho_{j}\leq 1-\epsilon. \nonumber\\
&& (\alpha_{j}(x)-\alpha^{*}(x))\eta_{j}(x)\leq 0,\nonumber\\
&& \forall x \in \mathcal{A},j \in \mathcal{B} .
\end{eqnarray}
Here, $\alpha_{j}(x)$ and $\alpha^{*}(x)$ are the uplink pathloss from the user at location $x$ to the $j$th BS and the uplink pathloss threshold for the user, respectively. $0<\epsilon<1$ is a small real number to ensure $\rho_{j}<1$. $\boldsymbol{\rho}=(\rho_{1}, \rho_{2}, \cdots, \rho_{|\mathcal{B}|})$, and %\begin{equation}
%w_{j}(\rho_{j})=(\frac{e^{\rho_{j}}}{e^{(1-\theta_{j})\rho_{j}+\theta_{j}\hat{\rho}_{j}}})^{\kappa}=e^{\kappa\theta_{j}(\rho_{j}-\hat{\rho}_{j})}
%\end{equation}
\begin{align}
\label{eq:latency_weight}
w_{j}(\rho_{j})&=e^{\kappa(\rho_{j}-\tau_{j})}\nonumber\\
&=e^{\kappa(\rho_{j}-(1-\theta_{j})\rho_{j}-\theta_{j}\hat{\rho}_{j})}\nonumber\\
&=e^{\kappa\theta_{j}(\rho_{j}-\hat{\rho}_{j})}
\end{align}
In the objective function, $w_{j}(\rho_{j})$ indicates the weight of the $j$th BS's latency indicator. If the $j$th BS has sufficient green energy ($\hat{\rho}_{j}\geq\rho_{j}$), ${0<w_{j}(\rho_{j})\leq 1}$; otherwise, $w_{j}(\rho_{j})>1$. This is because when the amount of available green energy in the $j$th BS is sufficient, the green traffic capacity, $\hat{\rho}_{j}$, is larger than $\rho_{j}$. Then, $\tau_{j}>\rho_{j}$ and $w_{j}<1$. With a large weight, the $j$th BS has a high priority in reducing its latency indicator while minimizing Eq.~(\ref{eq:object_network_goal}) as compared with the BSs having a small weight. Therefore, as compared with $w_{j}(\rho_{j})\leq 1$, $w_{j}(\rho_{j})>1$ enables the $j$th BS to achieve a smaller latency indicator.
Since
\begin{equation}
\label{eq:deritive_L}
\frac{dL(\rho_{j})}{d\rho_{j}}=\frac{\vartheta_{j}}{(1-\rho_{j})^2}>0,
\end{equation}
a smaller latency indicator means less traffic loads in the $j$th BS, which is desirable for saving on-grid power in the $j$th BS. Thus, introducing the weights for BSs' latency indicator in the objective function enables the green energy aware and traffic delivery latency aware user association. $\kappa$ is a parameter that further adjusts the value of the weight according to that of the traffic latency indicator and enables the network to control the trade-off between the on-grid power consumption and the average traffic delivery latency.

%\textcolor{red}{the intuition behind the model: close loop control figure}

\section{vGALA: a Green Energy and Latency Aware Load Balancing Scheme}
\label{sec:distributed_scheme}
In this section, we present the vGALA scheme and prove its properties. The vGALA scheme generally consists of three phases. The first phase is the initial user association and network measurement, during which the RANC collects network information, e.g., available green energy, traffic loads, and users' data rates. The second phase is the user association optimization, in which the RANC optimizes the user association and derives the corresponding BSs' operation statuses based on the information collected in the first phase. Here, a BS's operation status reflects the price for a user to access the BS. In the third phase, the user association is determined based on the optimized BSs' operation statuses and users' downlink data rates. The major optimization of the vGALA scheme is in the second phase. To be analytically tractable, we assume that (1) the RANC can successfully collect the network information from all BSs and users, and (2) the users' data rates do not change within one user association process. We will evaluate these assumptions in the next section where we discuss the practicality of the vGALA scheme.

\subsection{The vGALA user association scheme}
%
%In the first phase of the vGALA scheme, users and BSs measure and report their downlink data rates and green energy statuses to RANC, respectively.
Based on the collected network information, the RANC optimizes the user association and derives the optimal BS operation status. Leveraging the SoftRAN architecture, the RANC has a global view of the traffic loads and the availability of green energy in the network, to facilitate the user association optimization. However, owing to the large number of users and BSs, the user association algorithm if not well designed may be time consuming and incurs excessive delays. In order to efficiently optimize the user association, the vGALA scheme divides the user association algorithm into two parts: the user side algorithm and the BS side algorithm. The user side algorithm calculates the user's BS selection. The BS side algorithm updates the BS's operation status calculated based on the green traffic capacity and the traffic loads. Based on the updates, the user side algorithm re-calculates the BS selection. The user association algorithm iterates until it converges. After the convergence, the optimal BS operation status is obtained and the optimal user association is subsequently determined.

The information exchanges over the air interface between users and BSs may introduce additional communication overhead and incur extra power consumption. Leveraging cloud computing and virtualization, the vGALA scheme generates virtual users and virtual BSs (vBSs) in the RANC. The user side algorithm runs on virtual users while the BS side algorithm runs on vBSs. In this way, instead of exchanging information over the air interface, the virtual users and vBSs can iteratively update their information locally within the RANC. Here, the virtualization only virtualizes the computation resources for BSs and users rather than virtualzing all their functions.

\subsubsection{The user side algorithm}
We define the time interval between two consecutive BS selection updates as a time slot. At the beginning of the $k$th time slot, vBSs send their operation statuses to virtual users. Let
\begin{equation}
\psi(\boldsymbol{\rho})=\sum_{j \in \mathcal{B}}w_{j}(\rho_{j})L(\rho_{j}).
\end{equation}
The $j$th vBS's operation status in the $k$th time slot is defined as
\begin{align}
\label{eq:user_ass_message}
\phi_{j}(\rho_{j}(k))&=\frac{\partial \psi(\boldsymbol{\rho}(k))}{\partial {\rho}_{j}(k)}\nonumber\\
&=\frac{\vartheta_{j}e^{\kappa\theta_{j}(\rho_{j}(k)-\hat{\rho}_{j})}(\kappa\theta_{j}\rho_{j}(k)-\kappa\theta_{j}\rho_{j}(k)^{2}-1)}{(1-\rho_{j}(k))^2}.
\end{align}
Here, the $j$th vBS is mapped to the $j$th BS in the mobile network.

Let $\bar{\mathcal{B}}(x)=\{j|\alpha_{j}(x)\leq\alpha^{*}(x)\}$ be the set of BSs whose uplink pathloss is less than the user's pathloss threshold. Assign $r_{j}(x)=\zeta, \;\forall j\in\mathcal{B}\setminus\bar{\mathcal{B}}(x)$ where $\zeta$ is a very small positive number that approaches zero. This is equivalent to restricting the user from associating with the BSs outside $\bar{\mathcal{B}}(x)$.
Then, the BS selection rule for a user at location $x$ can be expressed as
\begin{equation}
\label{eq:bs_selection}
b^{k}(x)= \arg\max_{j \in \mathcal{B}}\frac{r_{j}(x)}{\phi_{j}(\rho_{j}(k))}
\end{equation}
Here, $b^{k}(x)$ is the index of the vBS selected by the virtual user at location $x$ in the $k$th time slot.
The pseudo code of the user side algorithm is shown in Alg. \ref{alg:usa_alg}. The computational complexity of the user side algorithm for an individual user is $O(|\mathcal{B}|)$.

\begin{algorithm}
\SetKwData{Left}{left}\SetKwData{This}{this}\SetKwData{Up}{up}
\SetKwFunction{Union}{Union}\SetKwFunction{FindCompress}{FindCompress}
\SetKwInOut{Input}{Input}\SetKwInOut{Output}{Output}
\Input{BSs' operation status: $\phi_{j}(\rho_{j}(k)), j\in\mathcal{B}$\;}
\Output{The BS selection: $b^{k}(x)$\;}
\nl Estimate the uplink pathloss: $\alpha_{j}(x)$\;
\nl Find $\bar{\mathcal{B}}(x)=\{j|\alpha_{j}(x)\leq\alpha^{*}(x)\}$\;
\nl Assign $r_{j}(x)=\zeta, \;\forall j\in\mathcal{B}\setminus\bar{\mathcal{B}}(x)$\;
\nl Find $b^{k}(x)= \arg\max_{j \in \mathcal{B}}\frac{r_{j}(x)}{\phi_{j}(\rho_{j}(k))}$\;
\caption{The User Side Algorithm\label{alg:usa_alg}}
\end{algorithm}

\subsubsection{The BS side algorithm}
Upon receiving vBSs' operation status updates, virtual users select vBSs according to the user side algorithm. Then, the coverage area of the $j$th vBS in the $k$th time slot is updated as
\begin{equation}
\label{eq:bs_cover_area}
\mathcal{\tilde{A}}_{j}(k)=\{x |j=b^{k}(x),\;\forall x \in \mathcal{A}\}
\end{equation}
Then, given $\boldsymbol{\rho}(k)=(\rho_{1}(k),\rho_{2}(k),\cdots,\rho_{|\mathcal{B}|}(k))$, $\boldsymbol{\theta}=({\theta_{1},\theta_{2},\cdots,\theta_{|\mathcal{B}|}})$, and $\boldsymbol{\hat{\rho}}=(\hat{\rho}_{1},\hat{\rho}_{2},\cdots,\hat{\rho}_{|\mathcal{B}|})$, the $j$th vBS's perceived traffic loads in the $k$th time slot is
\begin{equation}
\label{eq:update_traffic}
M_{j}(\boldsymbol{\rho}(k),\boldsymbol{\theta},\boldsymbol{\hat{\rho}})=\min{(\int_{x \in \mathcal{\tilde{A}}_{j}(k)}\varrho_{j}(x)dx,1-\epsilon)}.
\end{equation}
Since $\boldsymbol{\theta}$ and $\boldsymbol{\hat{\rho}}$ are assumed not to change within the duration of a user association process, $M_{j}(\boldsymbol{\rho}(k),\boldsymbol{\theta},\boldsymbol{\hat{\rho}})$ evolves based only on $\boldsymbol{\rho}(k)$. Thus, we use $M_{j}(\boldsymbol{\rho}(k))$ instead of $M_{j}(\boldsymbol{\rho}(k),\boldsymbol{\theta},\boldsymbol{\hat{\rho}})$ for simplicity in the following analysis.

After having derived the perceived traffic loads, the $j$th vBS updates its traffic loads as
\begin{equation}
\label{eq:traffic_updates}
\rho_{j}(k+1)=\delta(k)\rho_{j}(k)+(1-\delta(k))M_{j}(\boldsymbol{\rho}(k)).
\end{equation}
Here, $0\leq\delta(k)<1$ is a system parameter calculated by the RANC to enable
\begin{align}
\label{eq:back_line_search}
&\psi(\boldsymbol{\rho}(k+1))\nonumber\\
\;\;\;&\leq \psi(\boldsymbol{\rho}(k))+
\varsigma(1-\delta(k))\sum_{j\in\mathcal{B}} \phi_{j}(\rho_{j}(k))(M_{j}(\boldsymbol{\rho}(k))-\rho_{j}(k))
\end{align}
Here, $0<\varsigma<0.5$ is a constant.
In the $(k+1)$th time slot, the $j$th vBS's operation status is $\phi_{j}(\rho_{j}(k+1))$. The pseudo code of the BS sid algorithm is presented in Alg. \ref{alg:bsa_alg}. The computational complexity of the BS side algorithm is determined by the ``while'' loop whose running time depends on the convergence of $\psi(\boldsymbol{\rho}(k))$. When $\psi(\boldsymbol{\rho}(k))$ is closer to the optimal value, it may take longer time to find $\delta(k)$. In the following, we will analyze the convergence of the vGALA scheme, which reflects the computational complexity of the BS side algorithm.
\begin{algorithm}
%\DontPrintSemicolon
\SetKwData{Left}{left}\SetKwData{This}{this}\SetKwData{Up}{up}
\SetKwFunction{Union}{Union}\SetKwFunction{FindCompress}{FindCompress}
\SetKwInOut{Input}{Input}\SetKwInOut{Output}{Output}
\Input{Users' vBS selection: $b^{k}(x), \forall x\in\mathcal{A}$\;}
\Output{vBSs'operation status, $\phi_{j}(\rho_{j}(k+1)),\forall j\in\mathcal{B}$\;}
\nl vBSs measure their perceived traffic loads, $M_{j}(\boldsymbol{\rho}(k))$\;
\nl Assign $\delta(k)=0$\;
\nl \While {Eq.~(\ref{eq:back_line_search}) is not true}{
\nl $\delta(k)=1-\xi(1-\delta(k))$, here, $0<\xi<1$ is a real number\;}
\nl vBSs update their traffic loads: $\rho_{j}(k+1)=\delta(k)\rho_{j}(k)+(1-\delta(k))M_{j}(\boldsymbol{\rho}(k))$\;
\nl Calculate $\phi_{j}(\rho_{j}(k+1))$ based on $\rho_{j}(k+1)$, $\forall j\in\mathcal{B}$\;
\caption{The BS Side Algorithm\label{alg:bsa_alg}}
\end{algorithm}

\subsection{The convergence of vGALA}
\label{subsec:convergence}
In order to prove the convergence of vGALA, we first prove that the vBSs' traffic load vector converges. The feasible set for the UA problem is
\begin{align}
\label{eq:feasible_set}
\mathcal{F}=\lbrace &\boldsymbol{\rho}|\rho_{j}=\int_{x \in \mathcal{A}}\varrho_{j}(x)dx,\nonumber \\
&0\leq\rho_{j}\leq 1-\epsilon,\; \sum_{j\in\mathcal{B}}\eta_{j}(x)=1,\nonumber \\
& \eta_{j}(x)=\{0,1\},\;\forall j\in\mathcal{B},\;\forall x \in\mathcal{A}\rbrace
\end{align}
Since $\eta_{j}(x)=\{0,1\}$, $\mathcal{F}$ is not a convex set. Thus, the traffic updates in Eq.~(\ref{eq:traffic_updates}) cannot guarantee the updated traffic loads are in the feasible set. In order to show the convergence of vGALA, we first relax the constraint to let $0\leq\eta_{j}(x)\leq 1$ and then prove the traffic load vector converges to the traffic load vector that is in the feasible set. Define
\begin{align}
\mathcal{\tilde{F}}=\lbrace &\boldsymbol{\rho}|\rho_{j}=\int_{x \in \mathcal{A}}\varrho_{j}(x)dx,\nonumber\\
&0\leq\rho_{j}\leq 1-\epsilon,\; \sum_{j\in\mathcal{B}}\eta_{j}(x)=1,\nonumber\\
& 0\leq \eta_{j}(x)\leq 1,\;\forall j\in\mathcal{B},\;\forall x \in\mathcal{A}\rbrace
\end{align}
as the relaxed feasible set.

\begin{lemma}
\label{thm:feasible_set}
The relaxed feasible set $\mathcal{\tilde{F}}$ is a convex set.
\end{lemma}
\begin{proof}
The lemma is proved by showing that the set $\mathcal{\tilde{F}}$ contains any convex combination of the traffic load vector~$\boldsymbol{\rho}$.
\end{proof}
\begin{lemma}
\label{thm:convexity}
$\psi(\boldsymbol{\rho})$ is a strong convex function of $\boldsymbol{\rho}$ when $\boldsymbol{\rho}$ is defined in $\mathcal{\tilde{F}}$.
\end{lemma}
\begin{proof}
\label{prf:cvx}
The lemma is proved by showing $\bigtriangledown^{2}\psi(\boldsymbol{\rho})\succeq q\boldsymbol{I}$ where $q=4e^{-1}$ and $\boldsymbol{I}$ is an identity matrix.
\end{proof}
Let $\boldsymbol{M}(\boldsymbol{\rho})=\{M_{1}(\boldsymbol{\rho}),M_{2}(\boldsymbol{\rho}),\cdots,M_{|\mathcal{B}|}(\boldsymbol{\rho})\}$.
\begin{lemma}
\label{thm:descent_direction}
When $\boldsymbol{M}(\boldsymbol{\rho}(k)) \neq \boldsymbol{\rho}(k)$, $\boldsymbol{M}(\boldsymbol{\rho}(k))$ provides a descent direction of $\psi(\boldsymbol{\rho})$ at $\boldsymbol{\rho}(k)$.
\end{lemma}
\begin{proof}
\label{prf:descent_direction}
Since $\psi(\boldsymbol{\rho})$ is a convex function, proving the lemma is equivalent to prove
\begin{equation}
\langle\bigtriangledown{\psi(\boldsymbol{\rho})}|_{\boldsymbol{\rho}=\boldsymbol{\rho}(k)},\boldsymbol{M}(\boldsymbol{\rho}(k))-\boldsymbol{\rho}(k)\rangle<0.
\end{equation}
Let $\hat{\eta}_{j}(x)$ and $\eta_{j}(x)$ be the user association indication of the $j$th BS that result in the traffic load $M_{j}(\boldsymbol{\rho}(k))$ and $\rho_{j}(k)$, respectively.
\begin{eqnarray}
\label{eq:prf_descent_direction}
&&\langle\bigtriangledown{\psi(\boldsymbol{\rho})}|_{\boldsymbol{\rho}=\boldsymbol{\rho}(k)},\boldsymbol{M}(\boldsymbol{\rho}(k))-\boldsymbol{\rho}(k)\rangle\\
&&=\sum_{j\in\mathcal{B}}(M_{j}(\boldsymbol{\rho}(k))-\rho_{j}(k))\phi_{j}(\rho_{j}(k)) \nonumber\\
&&=\sum_{j\in\mathcal{B}}\frac{\int_{x \in \mathcal{A}}\lambda(x)\nu(x)(\hat{\eta}_{j}(x)-\eta_{j}(x))dx}{r_{j}(x)\phi^{-1}_{j}(\rho_{j}(k))} \nonumber\\
&&=\int_{x \in \mathcal{A}}\lambda(x)\nu(x)\sum_{j\in\mathcal{B}}\frac{\hat{\eta}_{j}(x)-\eta_{j}(x)}{r_{j}(x)\phi^{-1}_{j}(\rho_{j}(k))}dx. \nonumber
\end{eqnarray}
Since
%\begin{equation}
%\label{eq:eta_m}
%\hat{\eta}_{j}(x)=\mathbf{1}\{j=\arg\max_{l \in \mathcal{B}}r_{i,l}(x)(1-\rho_{j}(l))^{2}\phi_{j}(l)\}, \forall j \in \mathcal{B}, \forall x \in \mathcal{A},
%\end{equation}
\begin{equation}
\label{eq:eta_m}
\setlength{\nulldelimiterspace}{0pt}
\hat{\eta}_{j}(x)=\left\{\begin{IEEEeqnarraybox}[\relax][c]{ls}
1,
\; &for $j=b^{k}(x)$  \\
0, \; &for $otherwise$,%
\end{IEEEeqnarraybox}\right.
\end{equation}

\begin{equation}
\label{eq:sum_leq_0}
\sum_{j\in\mathcal{B}}\frac{\hat{\eta}_{j}(x)-\eta_{j}(x)}{r_{j}(x)\phi^{-1}_{j}(\rho_{j}(k))}\leq 0.
\end{equation}
Because $\boldsymbol{M}(\boldsymbol{\rho}(k)) \neq \boldsymbol{\rho}(k)$, there exists $j\in \mathcal{B}$ such that $\hat{\eta}_{j}(x)\neq\eta_{j}(x)$, $ x \in \mathcal{A}$. Hence,
\begin{equation}
\label{eq:sum_less_0}
\sum_{j\in\mathcal{B}}\frac{\hat{\eta}_{j}(x)-\eta_{j}(x)}{r_{j}(x)\phi^{-1}_{j}(\rho_{j}(k))}<0,
\end{equation}
and $\langle\bigtriangledown{\psi(\boldsymbol{\rho})}|_{\boldsymbol{\rho}=\boldsymbol{\rho}(k)},\boldsymbol{M}(\boldsymbol{\rho}(k))-\boldsymbol{\rho}(k)\rangle<0$.
\end{proof}
\begin{theorem}
\label{thm:alg_converge_thm}
The traffic load vector $\boldsymbol{\rho}$ converges to the traffic load vector $\boldsymbol{\rho}^{*}\in \mathcal{F}$.
\end{theorem}
\begin{proof}
\label{prf:alg_converge_thm}
Since $\sum_{j\in\mathcal{B}}(M_{j}(\boldsymbol{\rho}(k))-\rho_{j}(k))\phi_{j}(\rho_{j}(k))<0 $ when $\boldsymbol{M}(\boldsymbol{\rho}(k)) \neq \boldsymbol{\rho}(k)$, Alg. \ref{alg:bsa_alg} ensures $\psi(\boldsymbol{\rho}(k+1))\leq\psi(\boldsymbol{\rho}(k))$ in each time slot. Since $\psi(\boldsymbol{\rho})\geq 0$, $\psi(\boldsymbol{\rho})$ will converge. Let $\psi(\boldsymbol{\rho})$ converge to $\psi(\boldsymbol{\rho}^{*})$. Since
\begin{align}
\label{eq:alg_opt_1}
\boldsymbol{\rho}(k+1)&=\delta(k)\boldsymbol{\rho}(k)+(1-\delta(k))(\boldsymbol{M}\boldsymbol{\rho}(k)) \nonumber\\
&=\boldsymbol{\rho}(k) + (1-\delta(k))(\boldsymbol{M}(\boldsymbol{\rho}(k))-\boldsymbol{\rho}(k)),
\end{align}
$\boldsymbol{M}(\boldsymbol{\rho})$ and $\boldsymbol{\rho}$ will converge to $\boldsymbol{\rho}^{*}$. Because $\boldsymbol{M}(\boldsymbol{\rho}{*})$ is derived based on the user side algorithm where $\eta^{m}_{j}(x)=\{0,1\}, \forall j\in\mathcal{B},\; x\in\mathcal{A}$, $\boldsymbol{\rho}^{*}$ is in the feasible set~$\mathcal{F}$.
\end{proof}

\begin{corollary}
\label{thm:alg_converge_col}
The vBSs' operation status $\phi_{j}(\rho_{j}), \;\forall j \in \mathcal{B}$, converges to $\phi_{j}(\rho^{*}_{j})$.
\end{corollary}
\begin{proof}
Within the duration of a user association process, $\vartheta_{j}$, $\theta_{j}$, and $\hat{\rho}_{j}$ are constant. Thus, $\phi_{j}(\rho_{j})$ is only determined by $\rho_{j}$. Since $\rho_{j}$ converges to $\rho^{*}_{j}$, $\phi_{j}(\rho_{j})$ converges to $\phi_{j}(\rho^{*}_{j})$.
\end{proof}
Since $\psi(\boldsymbol{\rho})$ is a strong convex function, there exists $q>0$ and $Q>0$ such that $q\boldsymbol{I} \preceq\bigtriangledown^{2}\psi(\boldsymbol{\rho})\preceq Q\boldsymbol{I},\;\boldsymbol{\rho}\in\tilde{\mathcal{F}}$~\cite{Boyd:2004:CVX}. Denote $\psi(\boldsymbol{\rho}^{*})$ as the optimal solution. $\psi(\boldsymbol{\rho}(k+1))$ is said to be the $\epsilon$ suboptimal solution if $\psi(\boldsymbol{\rho}(k+1))-\psi(\boldsymbol{\rho}^{*})\leq \epsilon$ where $\epsilon>0$ is a small real number.
\begin{lemma}
\label{thm:converge_speed}
The number of iterations required to ensure $\psi(\boldsymbol{\rho}(k+1))-\psi(\boldsymbol{\rho}^{*})\leq \epsilon$ is at most equal to
\begin{equation}
\label{eq:num_iteration}
\frac{\log((\psi(\boldsymbol{\rho}(1))-\psi(\boldsymbol{\rho}^{*}))/\epsilon)}{\log{1/z}}
\end{equation}
where $z=1-\min\{2q\varsigma,2q\varsigma\xi/Q\}<1$ and $\boldsymbol{\rho}(1)$ is the initial traffic load vector.
\end{lemma}
\begin{proof}
The lemma is proved in Appendix \ref{app:proof_converge_speed}.
\end{proof}

Eq.~(\ref{eq:num_iteration}) indicates that $\psi(\boldsymbol{\rho})$ converges at least as fast as a geometric series. Such convergence is called linear convergence in the context of iterative numerical method~\cite{Boyd:2004:CVX}. The number of iterations required for $\psi(\boldsymbol{\rho})$ to converge depends on the gap between $\psi(\boldsymbol{\rho}(1))$ and $\psi(\boldsymbol{\rho}^{*})$, $\epsilon$, and $z$. Given the gap and the value of $\epsilon$, a smaller $z$ enables faster convergence. By properly selecting $\varsigma$ and $\xi$, we can reduce the value of $z$, and thus reduce the number of iterations required for the convergence. However, how to optimize the value of $\varsigma$ and $\xi$ is beyond the scope of this paper.

\subsection{The optimality of vGALA}
\label{subsec:optimal}
Since the vBSs' traffic load vector converges to $\boldsymbol{\rho}{*}$, we show that the corresponding user association minimizes $\psi(\boldsymbol{\rho})$.
\begin{theorem}
\label{thm:opt_vgala}
Suppose $\mathcal{F}$ is not empty and the traffic load vector converges to $\boldsymbol{\rho}{*}$, the user association corresponding to $\boldsymbol{\rho}{*}$ minimizes $\psi(\boldsymbol{\rho})$.
\end{theorem}
\begin{proof}
\label{prf:opt_vgala}
Denote $\boldsymbol{\eta}^{*}=\{\eta^{*}_{j}(x)|\eta^{*}_{j}(x)=\{0,1\},\; \forall j \in\mathcal{B}, \;\forall x\in \mathcal{A}\}$ and $\boldsymbol{\eta}=\{\eta_{j}(x)|\eta_{j}(x)=\{0,1\},\; \forall j \in\mathcal{B}, \;\forall x\in \mathcal{A}\}$ as the user association corresponding to $\boldsymbol{\rho}{*}$ and any other traffic load vector $\boldsymbol{\rho} \in \mathcal{F}$, respectively.

Let $\triangle \boldsymbol{\rho}^{*} =\boldsymbol{\rho}-\boldsymbol{\rho}^{*}$.
Since $\psi(\boldsymbol{\rho})$ is a convex function over $\boldsymbol{\rho}$, proving the theorem is equivalent to prove
\begin{equation}
\langle\bigtriangledown{\psi(\boldsymbol{\rho})}|_{\boldsymbol{\rho}=\boldsymbol{\rho}^{*}},\triangle \boldsymbol{\rho}^{*} \rangle\geq 0.
\end{equation}
\begin{eqnarray}
\label{eq:prf_opt_vgala}
&&\langle\bigtriangledown{\psi(\boldsymbol{\rho})}|_{\boldsymbol{\rho}=\boldsymbol{\rho}^{*}},\triangle \boldsymbol{\rho}^{*}\rangle\\
&&=\sum_{j\in\mathcal{B}}(\rho_{j}-\rho_{j}^{*})\phi_{j}(\rho^{*}_{j}) \nonumber\\
&&=\sum_{j\in\mathcal{B}}\frac{\int_{x \in \mathcal{A}}\lambda(x)\nu(x)(\eta_{j}(x)-\eta^{*}_{j}(x))dx}{r_{j}(x)\phi^{-1}_{j}(\rho^{*}_{j})} \nonumber\\
&&=\int_{x \in \mathcal{A}}\lambda(x)\nu(x)\sum_{j\in\mathcal{B}}\frac{\eta_{j}(x)-\eta^{*}_{j}(x)}{r_{j}(x)\phi^{-1}_{j}(\rho^{*}_{j})}dx. \nonumber
\end{eqnarray}
According to the user side algorithm,
\begin{equation}
\label{eq:eta_opt}
\setlength{\nulldelimiterspace}{0pt}
\eta^{*}_{j}(x)=\left\{\begin{IEEEeqnarraybox}[\relax][c]{ls}
1,
\; &for $j=\arg max_{i\in\mathcal{B}}\frac{r_{i}(x)}{\phi_{i}(\rho^{*}_{i})}$  \\
0, \; &for $otherwise$,%
\end{IEEEeqnarraybox}\right.
\end{equation}
Therefore,
\begin{equation}
\sum_{j\in\mathcal{B}}\frac{\eta^{*}_{j}(x)}{r_{j}(x)\phi^{-1}_{j}(\rho^{*}_{j})}\leq\sum_{j\in\mathcal{B}}\frac{\eta_{j}(x)}{r_{j}(x)\phi^{-1}_{j}(\rho^{*}_{j})}.
\end{equation}
Hence, $\langle\bigtriangledown{\psi(\boldsymbol{\rho})}|_{\boldsymbol{\rho}=\boldsymbol{\rho}^{*}},\triangle \boldsymbol{\rho}^{*} \rangle\geq 0$.
\end{proof}

%
%Then, $\psi(\boldsymbol{\rho}(k+1))<\psi(\boldsymbol{\rho}(k))$ until $\boldsymbol{\rho}(k+1)=\boldsymbol{\rho}(k)$. Therefore, the traffic load vector $\boldsymbol{\rho}$ converges to the optimal traffic load vector $\boldsymbol{\rho}^{*}$ that minimizes $\psi(\boldsymbol{\rho})$.
\subsection{The generalization of vGALA}
In determining the user association, the vGALA scheme strives for a balance between the green energy utilization and the network performance. In the problem formulation, $w_{j}(\rho_{j})$ and $L(\rho_{j})$ model the green energy utilization and the network performance, respectively. Since $w_{j}(\rho_{j})$ and $L(\rho_{j})$ are functions of the traffic load $\rho_{j}$, they are coupled by $\rho_{j}$. $L(\rho_{j})$ is a general latency indicator derived under the M/G/1 processor sharing queue model. In practical networks, traffic arrivals may follow arbitrary distributions rather than a Poisson distribution. In addition, the network operators may aim to represent the network performance with other metrics instead of the average traffic delivery latency. It is desirable that the vGALA framework can be applied to a collection of network performance models. Denote $f(\rho_{j})$ as a function of the traffic load $\rho_{j}$ that models the $j$th BS's performance. Define the user association problem with a generalized network performance model, $f(\rho_{j})$, as the UAG problem expressed as

\begin{eqnarray}
\label{eq:object_network_goal_general}
\min_{\boldsymbol{\rho}} && \sum_{j \in \mathcal{B}}w_{j}(\rho_{j})f(\rho_{j})\\
\label{eq:constraint_omge_general}
subject\; to: && 0\leq\rho_{j}\leq 1-\epsilon.
\end{eqnarray}

\begin{lemma}
\label{thm:gen_lemma}
If $f(\rho_{j})$ is positive, convex and non decreasing over $\rho_{j},\; \forall j\in\mathcal{B}$, $\tilde{\psi}(\boldsymbol{\rho})=\sum_{j\in\mathcal{B}}w_{j}(\rho_{j})f(\rho_{j})$ is convex over $\boldsymbol{\rho}\in \mathcal{\tilde{F}}$.
\end{lemma}
\begin{proof}
Since $f(\rho_{j})$ is positive, convex and non decreasing, $f(\rho_{j})>0$, $f^{\prime\prime}(\rho_{j})\geq 0$, and $f^{\prime}(\rho_{j})\geq 0$. Because $w_{j}^{\prime\prime}(\rho_{j})> 0$, $w_{j}^{\prime}(\rho_{j})> 0$, and $w_{j}(\rho_{j})> 0$,
\begin{align}
&\frac{\partial^{2}\sum_{j\in\mathcal{B}}w_{j}(\rho_{j})f(\rho_{j})}{\partial\rho_{j}^{2}}\nonumber\\
&=w_{j}^{\prime\prime}(\rho_{j})f(\rho_{j})+2w_{j}^{\prime}(\rho_{j})f^{\prime}(\rho_{j})+w_{j}(\rho_{j})f^{\prime\prime}(\rho_{j}) \nonumber\\
&\geq q.
\end{align}
Here, $q$ is a positive number. Let $\boldsymbol{I}$ be an identity matrix.
Since
\begin{equation}
\frac{\partial^{2}\sum_{j\in\mathcal{B}}w_{j}(\rho_{j})f(\rho_{j})}{\partial\rho_{j}\partial\rho_{i}}=0, \; \forall i\neq j,
\end{equation}
$\bigtriangledown^{2}\tilde{\psi}(\boldsymbol{\rho})\geq q\boldsymbol{I}$. Therefore, $\tilde{\psi}(\boldsymbol{\rho})$ is a strong convex function over $\boldsymbol{\rho}$, $\boldsymbol{\rho}\in \mathcal{\tilde{F}}$.
\end{proof}

%\begin{theorem}
%\label{thm:gen_vgala}
%The vGALA framework applies to optimize the user association with $f(\rho_{j})$ as the $j$th BS's network performance model if $f(\rho_{j})$ is positive, convex and non decreasing over $\rho_{j},\; \forall j\in\mathcal{B}$.
%\end{theorem}
\begin{theorem}
\label{thm:gen_vgala}
If the $j$th BS's network performance metric, $f(\rho_{j})$, is positive, convex and non decreasing over $\rho_{j},\; \forall j\in\mathcal{B}$, the UAG problem can be solve by the vGALA scheme.
\end{theorem}
\begin{proof}
\label{prf:gen_vgala}
In order to guarantee the convergence and the optimality of the vGALA scheme, $\tilde{\psi}(\boldsymbol{\rho})$ has to be strongly convex over $\boldsymbol{\rho}\in\mathcal{\tilde{F}}$. According to the above lemma, if $f(\rho_{j})$ is positive, convex and non decreasing, $\tilde{\psi}(\boldsymbol{\rho})$ is a strong convex function. Thus, the vGALA framework can be utilized to solve the UAG problem in which $f(\rho_{j})$ is the $j$th BS's network performance metric.
\end{proof}

\section{The Practicality of the vGALA Scheme}
\label{sec:tradeoff_admission}
In this section, we first present how to put the vGALA framework into practice and evaluate the assumptions made for developing the scheme. Then, we discuss two related issues on applying the vGALA scheme: the energy-latency trade-off and the admission control mechanism.

\subsection{Put into practice}
\label{subsec:put_into_practice}
In practical cellular networks, the traffic load balancing among BSs is usually triggered by network-level events, e.g., some BSs are congested while others are lightly loaded, rather than by user-level events, e.g., a few users' movement and data rate changes. Since a BS's traffic loads are determined by the average traffic load density of its coverage area, without considering green energy, it is reasonable to reduce a BS's coverage area to avoid traffic congestion if the traffic load density of the BS's coverage area is increasing. Therefore, a BS's traffic load can be derived based on the location-based traffic load density that reflect the traffic load density at a location. Thus, for practical implementation, the vGALA scheme collects the location-based traffic load density and the network green energy information in the first phase as shown in Fig.~\ref{fig:vgala_practic}. Given a specific location, it is realistic to assume that BSs' downlink data rates to users at the location are not changing during a traffic load balancing period. Notice that on modeling the traffic load in the UA problem, we differentiate users by their locations. Therefore, the vGALA scheme is compatible with the input of the location-based traffic load density and the location-based downlink data rates.

\begin{figure}
\centering
\includegraphics[scale=0.33]{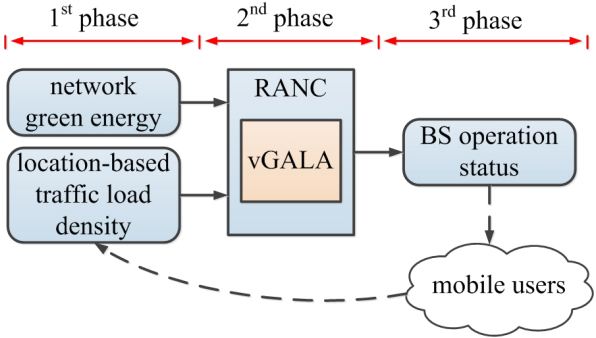}
\caption{The practical implementation of vGALA}
\label{fig:vgala_practic}
\vspace{-16pt}
\end{figure}

In the second phase, the vGALA scheme implemented in the RANC optimizes the user association and derives the optimal BS operation status based on the network information collected in the first phase. The optimization can be triggered either periodically or by some predefined events, e.g., a BS's traffic loads exceed a threshold or a BS's green energy utilization is lower than a threshold. What are the best strategies for triggering the traffic load balancing can be determined by network operators and is beyond the scope of this paper. The output of the second phase is the BS operation status, based on which the user association is determined in the third phase. In this phase, a user's BS association can be determined in either centralized or distributed fashion. In the first case, users send their data rate measurements to the RANC, and the RANC determines the users' BS associations based on the BS operation status and the users' date rates. In the second case, the RANC may simply let BSs broadcast their operation statuses, and based on which individual users decide their own BS associations. The users' BS selections may change the location-based traffic load density. Individual BSs translate the users' BS selections to location-based traffic load density and report it to the RANC.

In the vGALA scheme, the user association is optimized with the consideration of both the average traffic delivery latency and the green energy usage. From users' point of view (who may not care about the green energy usage), they may seek to maximize their performance and violate the BS selection rule in the vGALA scheme. However, the users, in fact, do not have any clue on maximizing their own QoS. According to Eq.~(\ref{eq:bs_selection}), a user's BS selection is based on both $r_{j}(x)$ and $\phi_{j}(\rho_{j})$. Here, $\phi_{j}(\rho_{j})$ is determined by both the $j$th BS's traffic loads and its available green energy. A user's average traffic delivery latency is determined by both the downlink data rate and the traffic loads of the associated BS. Since the users do not know the traffic loads of BSs, the users have no clue about which BS can provide them the best QoS. Simply selecting a BS with the largest $r_{j}(x)$ may lead the users to a highly congested BS and degrade the users' QoS. Thus, the users do not have obvious incentives to counterfeit their measurement reports.

\subsection{The energy-latency trade-off adaptation}
\label{subsec:energy_latency_trade}
The vGALA scheme provides two parameters for adapting the trade-off between the on-grid power consumption and the average traffic delivery latency. The parameters are $\theta$ and $\kappa$. $\theta$ is the energy-latency coefficient of a BS. It reflects individual BSs' operation strategies. A BS with a large $\theta$ ($\theta \rightarrow 1$) indicates that the BS is energy-sensitive. When a BS chooses a small $\theta$ ($\theta \rightarrow 0$), the BS is latency-sensitive. Therefore, by choosing the value of $\theta$, a BS adapts its sensitivity about the on-grid power consumption and the average traffic delivery latency. Hence, $\theta$ is chosen by individual BSs based on their operation strategies.

$\kappa$ is chosen by the RANC based on the global view of green energy status and the mobile traffic demands. Given $\theta$ and the available green energy, $w_{j}(\rho_{j})$ grows exponentially as the traffic demand increases. For a large $\kappa$, $w_{j}(\rho_{j})$ grows faster than it does with a small $\kappa$. This indicates that the vGALA scheme is more energy-sensitive when $\kappa$ is assigned a larger value. When $\kappa$ keeps increasing, the vGALA scheme will perform similarly as a solely energy-aware user association scheme. On the other hand, when $\kappa=0$, the vGALA scheme is a solely latency-aware user association scheme. In addition, since $0\leq \theta_{j}\leq 1$, $0\leq \theta_{j}\leq \kappa$. Thus, the value of $\kappa$ restricts the individual BSs' capability in adapting the energy-latency trade-off. The adaptation of $\kappa$ can be triggered by either green energy changes or the mobile traffic demand changes. For example, when the network experiences heavy traffic loads, the RANC will focus on balancing the traffic loads to reduce the network congestion. In this case, the RANC may choose a small $\kappa$ to give a high priority to the latency awareness in balancing the traffic loads. On the other hand, if the network experiences light traffic loads, the RANC may increase $\kappa$ to emphasize the green energy usage.

\begin{figure*}[ht]
\centering
\hspace*{\fill}
        \begin{subfigure}[b]{0.3\textwidth}
            \includegraphics[scale=0.3]{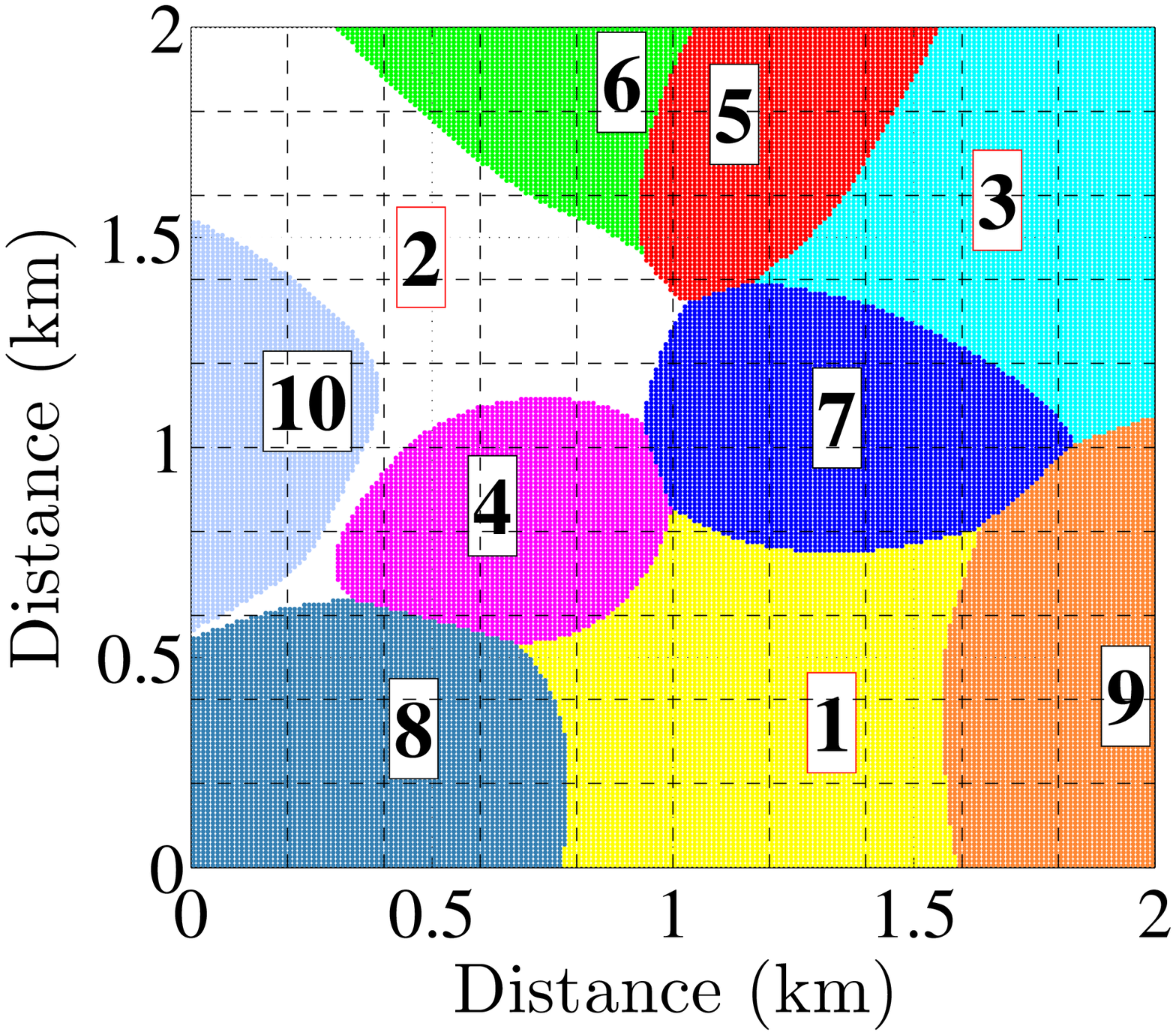}
            \caption{Green energy aware (GA).}
            \label{fig:ga_area}
        \end{subfigure}\hfill
        \begin{subfigure}[b]{0.3\textwidth}
            \includegraphics[scale=0.3]{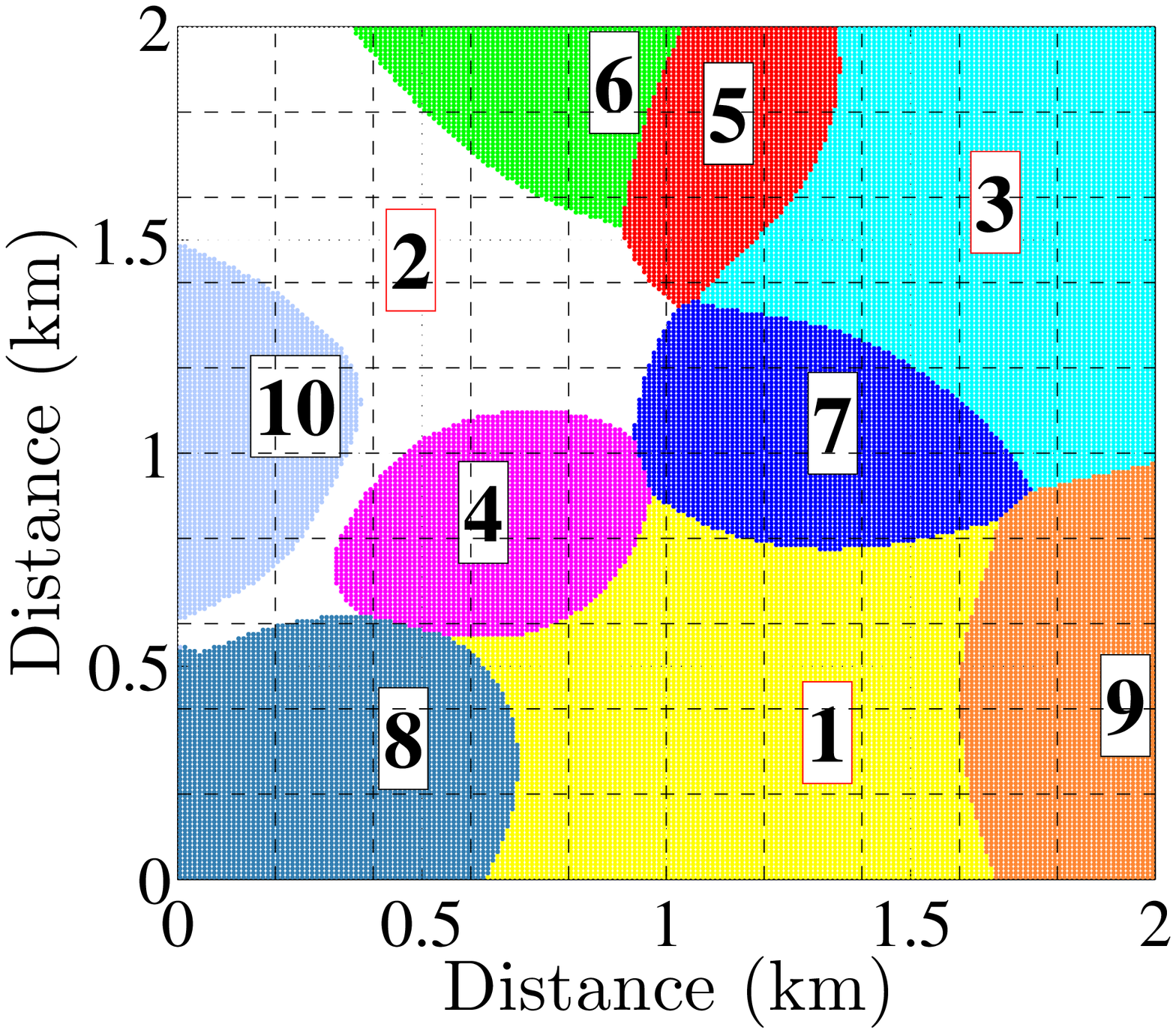}
            \caption{Latency aware (LA).}
            \label{fig:la_area}
        \end{subfigure}\hfill
        \begin{subfigure}[b]{0.3\textwidth}
            \includegraphics[scale=0.3]{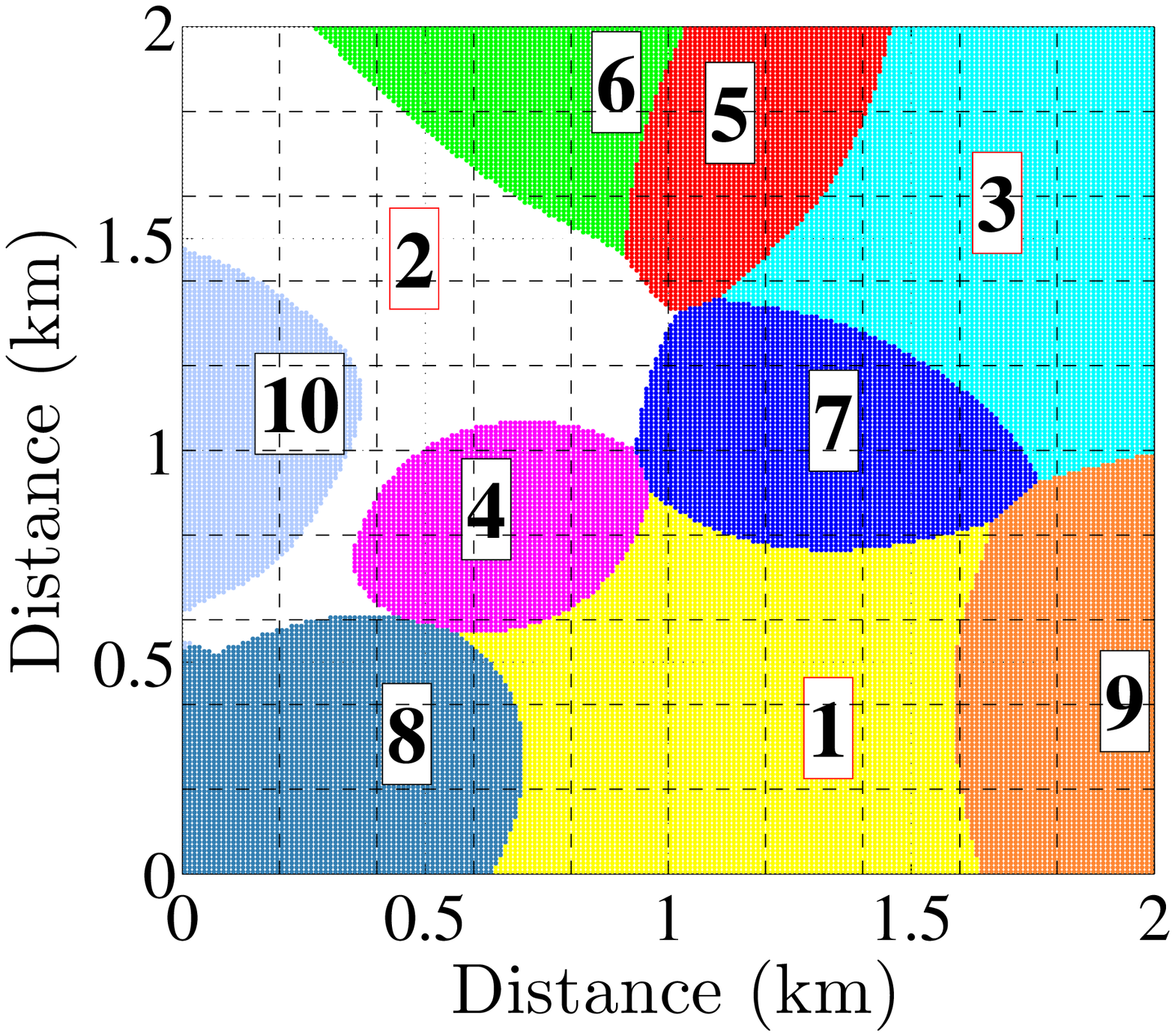}
            \caption{vGALA ($\theta=0.8$, $\kappa=4$).}
            \label{fig:vgala_area}
        \end{subfigure}\hfill
    \caption{%
       The coverage areas of different user association schemes.
     }%
   \label{fig:sim_1_area_comp}
   \vspace{-16pt}
\end{figure*}

\subsection{Admission control mechanism}
The necessary condition for the convergence and optimality of the vGALA scheme is that the UA problem is feasible. In other words, the BSs' traffic loads should be within the feasible set defined in Eq.~(\ref{eq:feasible_set}). When the traffic loads are beyond the network capacity, the UA problem is no longer feasible. As a result, the properties of the vGALA scheme will not hold. Therefore, the admission control mechanism is necessary for the vGALA scheme to ensure the feasibility of the UA problem. Thus, the purpose of proposing a simple admission control mechanism is to ensure that the vGALA scheme works even under very heavy traffic load condition (when the UA problem is not feasible) rather than to reduce either the energy consumption or average traffic delivery latency of the network.

Denote $\mu(x)$ as the admission control coefficient for a user located at $x$. $0\leq\mu(x)\leq 1$ indicates the probability that a user at location $x$ is admitted to the network. The RANC assigns $\mu(x)$ to a user at location $x$. $\mu(x)$ does not depend on the user's BS selection. In other words, no matter which BS is selected by a user, the user's admission control coefficient does not change. Thus, integrating admission control mechanism does not change the BS selection rule of the users. The coverage area of a BS, e.g., $\mathcal{\tilde{A}}_{j}(k)$, is still calculated by Eq.~(\ref{eq:bs_cover_area}). Owing to the admission control, the traffic load measurement in the $j$th vBS is revised as
\begin{equation}
\label{eq:ad_update_traffic}
M_{j}(\boldsymbol{\rho}(k))=\min{(\int_{x \in \mathcal{\tilde{A}}_{j}(k)}\mu(x)\varrho_{j}(x)dx,1-\epsilon)}.
\end{equation}
The vBS updates its traffic loads based on Eq.~(\ref{eq:traffic_updates}).

With the admission control, the RANC is able to restrict the traffic loads in the network to ensure the UA problem being feasible. The relaxed feasible set for the UA problem with admission control is
\begin{align}
\label{eq:relax_feasible_uac}
\mathcal{\hat{F}}=\lbrace &\boldsymbol{\rho}|\rho_{j}=\int_{x \in \mathcal{A}}\mu(x)\varrho_{j}(x)dx,\nonumber\\
&0\leq\rho_{j}\leq 1-\epsilon,\; \sum_{j\in\mathcal{B}}\eta_{j}(x)=1,\nonumber\\
& 0\leq \eta_{j}(x)\leq 1,\;\forall j\in\mathcal{B},\;\forall x \in\mathcal{A}\rbrace
\end{align}
Since $0\leq\mu(x)\leq 1$ is a constant, Lemma \ref{thm:feasible_set} still holds, which means that $\mathcal{\hat{F}}$ is a convex set. Integrating admission control does not change the objective function of the UA problem. Thus, Lemma \ref{thm:convexity} also holds. By applying the similar analysis presented in Sections \ref{subsec:convergence} and \ref{subsec:optimal}, we can prove that the vGALA scheme still enables the convergence of the traffic loads and obtains the optimal solution to the UA problem with the admission control.

%
%
%To guarantee the optimality and convergence of eGALA scheme, we introduce a two-level admission control scheme in BSs to ensure 1) the UA problem is feasible and 2) users have no incentives to deviate from the BS selection rule.
%
%
%first level -- hard admission control --- security issue --- selfish user
%
%second level -- probabilistic admission control --- control the traffic load

\section{Simulation Results}
\label{sec:simulation}

\begin{table}[ht]
\caption{Channel Model and Parameters}
\centering
\begin{tabular}{l||l}
\hline
Parameters & Value\\
\hline
$PL_{MBS}$ (dB) & $PL_{MBS}=128.1+37.6\log_{10}(d)$\\
$PL_{SCBS}$ (dB) & $PL_{SCBS}=38+10\log_{10}(d)$ \\
Rayleigh fading & 9 $dB$\\
Shadowing fading & 5 $dB$ \\
Antenna gain & 15 $dB$\\
Noise power level & -174 $dBm$ \\
Receiver sensitivity & -123 $dBm$ \\
\hline
\end{tabular}
\label{table:sim_parameters}
\end{table}

\begin{figure*}[ht]
\centering
\hspace*{\fill}
    \begin{subfigure}[b]{0.5\textwidth}
            \includegraphics[width=2.6in]{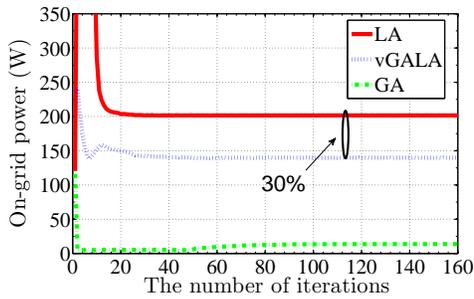}
            \caption{The on-grid power consumption.}
            \label{fig:sim_1_power}
    \end{subfigure}%
    \begin{subfigure}[b]{0.5\textwidth}
            \includegraphics[width=2.6in]{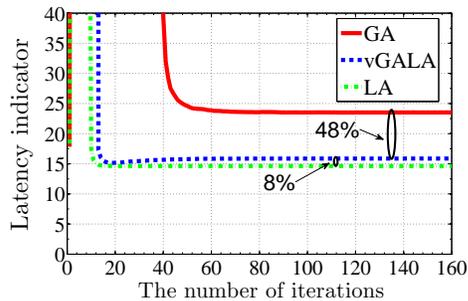}
            \caption{The average traffic delivery latency.}
            \label{fig:sim_1_latency}
   \end{subfigure}\hfill
   \caption{%
       The comparison of different user association scheme ($\theta=0.8$, $\kappa=4$).
     }%
   \label{fig:sim_1_power_latency}
\end{figure*}

We set up system level simulations to investigate the performance of the vGALA scheme for the downlink traffic load balancing in HetNet. In the simulation, three MBSs and seven SCBSs are randomly deployed in a $2000 m \times 2000 m$ area. The traffic arrival in the area follows the Poisson point process with the average arrival rate equaling to 200. The traffic size per arrival is 250~$kbits$. The area is divided into 40000 locations with each location representing a $10 m \times 10 m$ area. The location-based traffic load density is calculated based on the traffic model.
The static power consumption of the MBS and the SCBS are 750~$W$ and 37~$W$, respectively~\cite{Auer:2011:HME}. The load-power coefficient of the MBS and the SCBS are 500 and 4, respectively~\cite{Auer:2011:HME}. The solar cell power efficiency is $17.4\%$~\cite{HIT:Photo}. We assume that the weather condition is the standard condition which specifies a temperature of 25 $^{o}C$, an irradiance of 1000 $W/m^{2}$, and an air mass of 1.5 spectrum. Thus, the green energy generation rate is 174 $W/m^{2}$. The solar panel sizes are randomly selected but ensure the green power generation capacity of MBSs from 750~$W$ to 1300~$W$ while that of SCBS from 37~$W$ to 48~$W$. BSs' energy-latency coefficients are set to be the same.
The total bandwidth is 20~$MHz$ in which 10~$MHz$ is exclusively used by MBSs and the other 10~$MHz$ is allocated to SCBSs. The frequency reuse factor for each system (MBSs and SCBSs) is one. The channel propagation model is based on COST 231 Walfisch-Ikegami~\cite{COST231}. The model and parameters are summarized in Table \ref{table:sim_parameters}. Here, $PL_{MBS}$ and $PL_{SCBS}$ are the path loss between the users and MBSs and SCBSs, respectively. $d$ is the distance between users and BSs.

\subsection{Performance comparison}
We compare the vGALA scheme with a green energy aware (GA) user association scheme and a latency aware (LA) user association scheme. The GA scheme solves the green energy aware problem (GAP) formulated as
\begin{eqnarray}
\label{eq:obj_gap}
\min_{\boldsymbol{\rho}} && \sum_{j \in \mathcal{B}}\max(\rho_{j}-e_{j},0)\\
\label{eq:constraint_gap}
subject\; to: && 0\leq\rho_{j}\leq 1-\epsilon.
\end{eqnarray}
The LA scheme solves the latency aware problem (LAP) as
\begin{eqnarray}
\label{eq:obj_lap}
\min_{\boldsymbol{\rho}} && \sum_{j \in \mathcal{B}}L(\rho_{j})\\
\label{eq:constraint_lap}
subject\; to: && 0\leq\rho_{j}\leq 1-\epsilon.
\end{eqnarray}

\begin{figure*}
\centering
\hspace*{\fill}
    \begin{subfigure}[b]{0.5\textwidth}
            \includegraphics[width=2.6in]{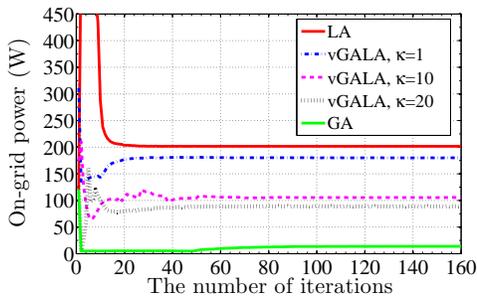}
            \caption{The on-grid power consumption.}
            \label{fig:sim_2_power}
    \end{subfigure}%
    \begin{subfigure}[b]{0.5\textwidth}
            \includegraphics[width=2.6in]{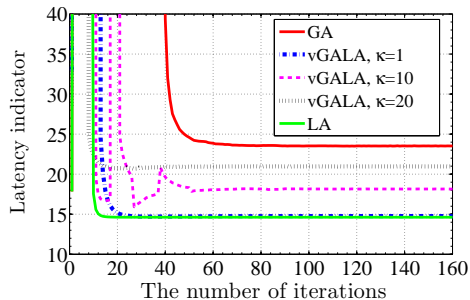}
            \caption{The average traffic delivery latency.}
            \label{fig:sim_2_latency}
   \end{subfigure}\hfill
    \caption{%
       The performance of vGALA with various $\kappa$ ($\theta=1$).
     }%
   \label{fig:sim_2_power_latency}
   \vspace{-16pt}
\end{figure*}

As shown in Figs.~\ref{fig:sim_1_area_comp}, different user association schemes result in different traffic load distribution among BSs. In the figure, the coverage areas of different BSs are filled with different colors\footnote{The white color indicates the coverage area of the second BS.}. A larger coverage area indicates the BS serves more traffic loads. The first, second and third BSs are MBSs and the other BSs are SCBSs. Taking the coverage area of the 5th BS as an example, as compared with the GA scheme (Fig. \ref{fig:ga_area}), the LA scheme significantly reduces the BS's coverage area as shown in Fig. \ref{fig:la_area}. The 5th BS has sufficient green energy. Therefore, the GA scheme will redirect more traffic loads to the BS to minimize the on-grid power consumption. The LA scheme, which does not consider the energy usage, balances the traffic loads among BSs to minimize the average traffic delivery latency. As a result, the LA scheme limits the traffic loads in the BS. Considering both the power consumption and the average traffic delivery latency, the vGALA scheme slightly reduces the BS's coverage area as shown in Fig. \ref{fig:vgala_area} to obtain a trade-off between the on-grid power consumption and the average traffic delivery latency.

Fig. \ref{fig:sim_1_power_latency} shows the trade-off achieved by the vGALA scheme between the on-grid energy consumption and the average traffic delivery latency. Fig. \ref{fig:sim_1_power} shows the on-grid power consumption of the LA, the vGALA, and the GA schemes, respectively. As compared with the LA scheme, the vGALA scheme consumes 30\% less on-grid power. Fig. \ref{fig:sim_1_latency} shows that the average traffic delivery latency of the vGALA scheme is only 8\% more than that of the LA scheme. While the GA scheme significantly reduces the on-grid power consumption, it increases the traffic delivery latency by about 48\% percent as compared with the vGALA scheme. Here, the latency indicator equals to $\sum_{j \in \mathcal{B}}L(\rho_{j})$. The above observation indicates that the vGALA scheme achieves a preferable trade-off: saving 30\% on-grid power at the cost of 8\% increase in the average traffic delivery latency. In addition, as shown in Fig. \ref{fig:sim_1_power_latency}, the vGALA scheme requires about 60 iterations to converge to the optimal solution. On the one hand, it proves that the vGALA scheme converges fast. On the other hand, it indicates that the vGALA scheme avoids the communication overhead over the air interface by virtualizing users and BSs in the RANC to simulate the interactions between users and BSs.

\subsection{Performance adaptation}
The trade-off between the on-grid power consumption and the average traffic delivery latency can be adapted by adjusting $\kappa$ and $\theta$ in the vGALA scheme. Fig. \ref{fig:sim_2_power_latency} shows the performance of the vGALA scheme with different $\kappa$. By varying $\kappa$, the vGALA scheme may act as the LA scheme when $\kappa \rightarrow 0$ and performs like the GA scheme when $\kappa \rightarrow \infty$. As shown in Fig. \ref{fig:sim_3_power_latency}, given $\kappa$, adjusting $\theta$ has a limited performance adaptation. In other words, $\kappa$ defines a performance adaptation range and adjusting $\theta$ can only adapt the performance within the range. As discussed in Section \ref{subsec:energy_latency_trade}, the selection of $\theta$ is determined by the operation strategies of BSs while the value of $\kappa$ is chosen based on network conditions, e.g., the traffic load intensity and the available green energy. However, how to optimize these values is beyond the scope of this paper.

\begin{figure*}
\centering
\hspace*{\fill}
    \begin{subfigure}[b]{0.5\textwidth}
            \includegraphics[width=2.6in]{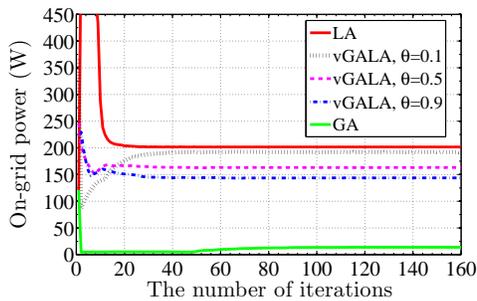}
            \caption{The on-grid power consumption.}
            \label{fig:sim_3_power}
    \end{subfigure}%
    \begin{subfigure}[b]{0.5\textwidth}
            \includegraphics[width=2.6in]{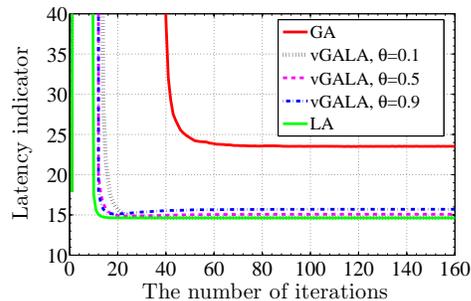}
            \caption{The average traffic delivery latency.}
            \label{fig:sim_3_latency}
   \end{subfigure}\hfill
    \caption{%
       The performance of vGALA with various $\theta$ ($\kappa=4$).
     }%
   \label{fig:sim_3_power_latency}
\end{figure*}

\begin{figure*}
\centering
\hspace*{\fill}
    \begin{subfigure}[b]{0.5\textwidth}
            \includegraphics[width=2.6in]{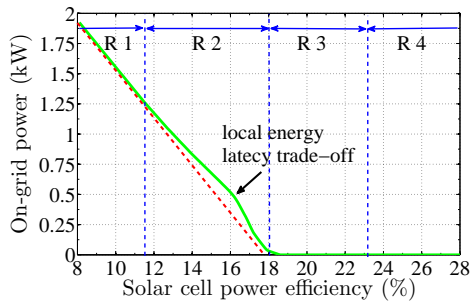}
            \caption{The on-grid power consumption.}
            \label{fig:sim_4_power}
    \end{subfigure}%
    \begin{subfigure}[b]{0.5\textwidth}
            \includegraphics[width=2.6in]{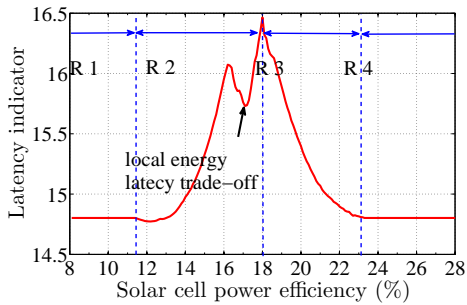}
            \caption{The average traffic delivery latency.}
            \label{fig:sim_4_latency}
   \end{subfigure}\hfill
       \caption{%
       The performance of vGALA versus solar cell power efficiency ($\theta=0.8$, $\kappa=4$).
     }%
   \label{fig:sim_4_power_latency}
   \vspace{-16pt}
\end{figure*}

\subsection{Green energy generation rate evaluation}
The amount of green energy in BSs impacts the performance of the vGALA scheme. In Fig. \ref{fig:sim_4_power_latency}, the x-axis is the solar cell power efficiency. As the solar cell power efficiency enhances, the amount of green energy in BSs will increase. As shown in Fig. \ref{fig:sim_4_power}, the on-grid power consumption of BSs decreases as the solar cell power efficiency increases. This is because more green energy is available in BSs. With the increase of the solar cell power efficiency, the performance on the average traffic delivery latency can be divided into four regions as shown in Fig. \ref{fig:sim_4_latency}. In the first region (R1), all BSs do not have sufficient green energy to offset their static power consumption. As a result, BSs' green traffic capacities are zero. In this condition, the vGALA scheme performs like the LA scheme. In the second region (R2), the green traffic capacities of BSs start to impact the traffic load balancing. The traffic loads will be directed to BSs that have sufficient green energy. Meanwhile, the vGALA scheme avoids to excessively increase the average traffic delivery latency. In the region, green energy is not sufficient in the network. Thus, the major strategy is to trade the average traffic delivery latency for saving on-grid power. However, as the solar power efficiency increases, some BSs may have sufficient green energy and they start trading their green energy for reducing the average traffic delivery latency in the network (the solar power efficiency falls between 16\% and 17\%). This event reflects the local energy-latency trade-off among several BSs. In the third region (R3), as the solar cell power efficiency further increases, the traffic load balancing becomes more flexible with respect to the green energy constraint, which enables the vGALA scheme to further reduce the average traffic delivery latency. In both region R2 and R3, the vGALA scheme determines the trade-off between the on-grid power consumption and the average traffic delivery latency. In the fourth region (R4), all BSs have sufficient green energy to operate with full traffic loads. In other words, the green traffic capacities of all the BSs equal to one. Thus, green energy is no longer a concern in balancing the traffic load and the vGALA scheme acts as the LA scheme.

\begin{figure*}
\centering
\hspace*{\fill}
    \begin{subfigure}[b]{0.5\textwidth}
            \includegraphics[width=2.5in]{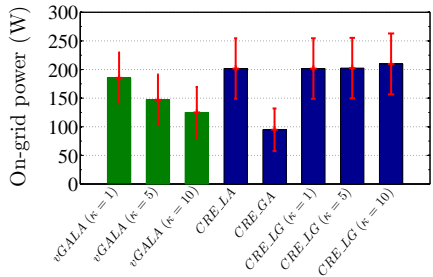}
            \caption{The on-grid power consumption.}
            \label{fig:sim_6_power}
    \end{subfigure}%
    \begin{subfigure}[b]{0.5\textwidth}
            \includegraphics[width=2.5in]{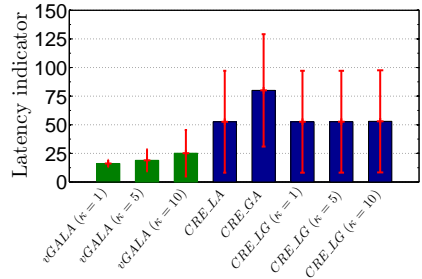}
            \caption{The average traffic delivery latency.}
            \label{fig:sim_6_latency}
   \end{subfigure}\hfill
       \caption{%
       The performance of vGALA versus CRE ($\theta=0.8$).
     }%
   \label{fig:sim_6_power_latency}
\vspace{-16pt}
\end{figure*}

\subsection{Practicality evaluation}
The cell range expansion (CRE) approach is one of the most practical traffic load balancing approach and has been proven to have similar performance as optimal traffic balancing schemes in term of maximizing network utilities~\cite{Andrews:2014:AOLB,Ye:2013:UAL}.
This simulation evaluates the traffic balancing performance of the vGALA scheme and the CRE approach.
%The performance evaluation is based on practical implementations.
For the vGALA scheme, the simulation follows Section \ref{subsec:put_into_practice} to obtain the optimal BS operation status based on the location-based traffic load density of the coverage area and the available green energy. We adopt the two-tier data rate bias approach as the CRE approach and assume that BSs in the same tier have the data rate bias. In the simulation, MBSs are in the first tier while SCBSs are in the second tier. In the data rate bias approach, a user selects the BS to maximize the biased data rate.
\begin{equation}
\label{eq:cre_alg}
b(x)=\arg\max_{j\in\mathcal{B}} Z_{j}r_{j}(x).
\end{equation}
Here, $b(x)$ and $Z_{j}$ are the index of the selected BS and the data rate bias of the $j$th BS. The data rate bias of a MBS is one. The data rate biases are selected for SCBSs to minimize (1) the average traffic delivery latency, (2) the overall on-grid power consumption, and (3) $\psi(\boldsymbol{\rho})$. Define these data rate biases as (1) CRE\_LA (latency-aware), (2) CRE\_GA (green energy-aware), and (3) CRE\_LG (latency and green energy-aware), respectively.

In the simulation, the BS operation status and the data rate biases are calculated based on the location-based traffic load density generated in previous simulations. We randomly generate users' locations using Poisson point process\footnote{The Poisson point process is the same as the Poisson point process used to generate the location-based traffic load density.} with average rate equalling to 200 in the area. The average traffic size per user is 250~$kbits$. We run the simulation 10000 times to evaluate the performance of different approaches in terms of the average traffic delivery latency and the average on-grid power consumption. As shown in Fig. \ref{fig:sim_6_power_latency}, CRE\_GA achieves the minimum on-grid energy consumption among all the schemes. However, the average traffic delivery latency of CRE\_GA is significantly larger than other schemes. As compared with CRE\_LA and CRE\_LG, the vGALA scheme not only saves the on-grid energy consumption but also reduces the average traffic delivery latency. For the vGALA scheme, when $\kappa$ increases, the scheme is to gradually prioritize saving on-grid energy in balancing the traffic loads, as shown in Fig. \ref{fig:sim_6_power}, at the cost of a small increase of the average traffic delivery latency as shown in Fig. \ref{fig:sim_6_latency}. For the CRE\_LG scheme, increasing $\kappa$ does not effectively adjust the energy-latency trade-off as the vGALA scheme does. This indicates the tier-based data rate bias approach may not perform well on jointly optimizing the utilization of green energy and the network utilities.

\section{Conclusion}
\label{sec:conclusion}
In this paper, we have proposed a traffic load balancing framework referred to as vGALA. During the procedure of establishing user association, the vGALA scheme not only considers the network performance, e.g., the average traffic delivery latency, but also adapts to the availability of green energy. Various properties, in particular, convergence of vGALA, have been proven. The vGALA scheme reduces the on-grid power consumption with a little sacrifice of the average traffic delivery latency. The trade-off between the network performance and the on-grid power consumption is adjustable in individual BSs and controllable by the radio access network controller. The vGALA scheme includes both the user side algorithm and the BS side algorithm. To avoid the extra communication overheads, the vGALA scheme, leveraging the SoftRAN architecture, introduces virtual users and vBSs to simulate the interactions between users and BSs thus significantly reducing the information exchanges over the air interface. The extensive simulation results have validated the performance and the practicality of the vGALA scheme.

\appendices
\section{Proof of Lemma \ref{thm:converge_speed}}
\label{app:proof_converge_speed}
Let $\bigtriangleup \boldsymbol{\rho}(k)=\boldsymbol{M}(\boldsymbol{\rho}(k))-\boldsymbol{\rho}(k)$. The termination condition of the BS side algorithm (Alg. \ref{alg:bsa_alg}) can be expressed as
\begin{align}
\label{eq:termination_a}
&\psi(\boldsymbol{\rho}(k+1))\nonumber\\
\;\;\;&\leq \psi(\boldsymbol{\rho}(k))+
\varsigma(1-\delta(k))\bigtriangledown\psi(\boldsymbol{\rho})^{\top}\bigtriangleup \boldsymbol{\rho}(k)
\end{align}
Since $\bigtriangleup \boldsymbol{\rho}(k)$ is a descent direction of $\psi(\boldsymbol{\rho}(k))$, $\bigtriangleup \boldsymbol{\rho}(k)$ can be replaced by $-\bigtriangledown\psi(\boldsymbol{\rho})$. Thus, the termination condition of Alg. \ref{alg:bsa_alg} can be rewritten as
\begin{align}
\label{eq:termination_b}
\psi(\boldsymbol{\rho}(k+1))\leq \psi(\boldsymbol{\rho}(k))-\varsigma(1-\delta(k))\Vert\bigtriangledown\psi(\boldsymbol{\rho})\Vert^{2}_{2}
\end{align}
Next, we will prove that the termination condition is satisfied whenever $0\leq 1-\delta(k)\leq 1/Q$. Since $\psi(\boldsymbol{\rho})\preceq Q\boldsymbol{I}$, we can derive, according to~\cite{Boyd:2004:CVX},
\begin{align}
&\psi(\boldsymbol{\rho}(k+1))\leq \nonumber\\ &\psi(\boldsymbol{\rho}(k))+(\frac{(1-\delta(k))Q}{2}-1)(1-\delta(k))\Vert\bigtriangledown\psi(\boldsymbol{\rho})\Vert^{2}_{2}.
\end{align}
When $0\leq 1-\delta(k)\leq 1/Q$, $\frac{(1-\delta(k))Q}{2}-1 \leq -1/2$. Therefore,
\begin{align}
&\psi(\boldsymbol{\rho}(k+1))\leq \nonumber\\ &\psi(\boldsymbol{\rho}(k))-\frac{(1-\delta(k))}{2}\Vert\bigtriangledown\psi(\boldsymbol{\rho})\Vert^{2}_{2}.
\end{align}
Since $0<\varsigma<0.5$, $-\frac{(1-\delta(k))}{2}\leq -(1-\delta(k))\varsigma$. Thus, we have
\begin{align}
\psi(\boldsymbol{\rho}(k+1))\leq \psi(\boldsymbol{\rho}(k))-(1-\delta(k))\varsigma\Vert\bigtriangledown\psi(\boldsymbol{\rho})\Vert^{2}_{2},
\end{align}
which satisfies the termination condition of Alg. \ref{alg:bsa_alg}. Therefore, Alg. \ref{alg:bsa_alg} terminates either with $\delta(k)=0$ or $(1-\delta(k))$ equaling to a value that is larger than $\xi/Q$.

In the first case ($\delta(k)=0$), we have
\begin{align}
\psi(\boldsymbol{\rho}(k+1))\leq \psi(\boldsymbol{\rho}(k))-\varsigma\Vert\bigtriangledown\psi(\boldsymbol{\rho})\Vert^{2}_{2}.
\end{align}
In the second case ($(1-\delta(k))\geq \xi/Q$), we can derive that
\begin{align}
\psi(\boldsymbol{\rho}(k+1))\leq \psi(\boldsymbol{\rho}(k))-\varsigma\xi/Q\Vert\bigtriangledown\psi(\boldsymbol{\rho})\Vert^{2}_{2}.
\end{align}
Thus,
\begin{align}
\psi(\boldsymbol{\rho}(k+1))\leq \psi(\boldsymbol{\rho}(k))-\min\{\varsigma,\varsigma\xi/Q\}\Vert\bigtriangledown\psi(\boldsymbol{\rho})\Vert^{2}_{2}.
\end{align}
Subtracting $\psi(\boldsymbol{\rho}^{*})$ from both side, we have
\begin{align}
&\psi(\boldsymbol{\rho}(k+1))-\psi(\boldsymbol{\rho}^{*})\leq \nonumber\\ &\psi(\boldsymbol{\rho}(k))-\psi(\boldsymbol{\rho}^{*})-\min\{\varsigma,\varsigma\xi/Q\}\Vert\bigtriangledown\psi(\boldsymbol{\rho})\Vert^{2}_{2}.
\end{align}
Since $q\boldsymbol{I} \preceq\bigtriangledown^{2}\psi(\boldsymbol{\rho})$, according to~\cite{Boyd:2004:CVX},
\begin{align}
\Vert\bigtriangledown\psi(\boldsymbol{\rho}(k))\Vert^{2}_{2}\geq 2q(\psi(\boldsymbol{\rho}(k))-\psi(\boldsymbol{\rho}^{*})).
\end{align}
Combining these together, we can derive that
\begin{align}
&\psi(\boldsymbol{\rho}(k+1))-\psi(\boldsymbol{\rho}^{*})\leq \nonumber\\
&(1-\min\{2q\varsigma,2q\varsigma\xi/Q\})(\psi(\boldsymbol{\rho}(k))-\psi(\boldsymbol{\rho}^{*})).
\end{align}
Let $z=1-\min\{2q\varsigma,2q\varsigma\xi/Q\}$ and apply the inequality recursively, we find that
\begin{align}
&\psi(\boldsymbol{\rho}(k+1))-\psi(\boldsymbol{\rho}^{*})\leq z^{k}(\psi(\boldsymbol{\rho}(1))-\psi(\boldsymbol{\rho}^{*})).
\end{align}
Let $z^{k}(\psi(\boldsymbol{\rho}(1))-\psi(\boldsymbol{\rho}^{*}))=\epsilon$; we derive that the number of iteration required to achieve $\epsilon$ optimality is
\begin{equation}
k=\frac{\log((\psi(\boldsymbol{\rho}(1))-\psi(\boldsymbol{\rho}^{*}))/\epsilon)}{\log{1/z}}.
\end{equation}
\bibliographystyle{IEEEtran}
\bibliography{mybib}

\end{document}